\documentclass[12pt,onecolumn]{IEEEtran}
\usepackage{amsmath,amssymb,amsfonts}
\usepackage[ruled,vlined,linesnumbered]{algorithm2e}
\usepackage{array}
\usepackage[caption=false,font=normalsize,labelfont=sf,textfont=sf]{subfig}
\usepackage{textcomp}
\usepackage{stfloats}
\usepackage{url}
\usepackage{verbatim}
\usepackage{graphicx}
\usepackage{cite}

\usepackage{microtype}  
\usepackage{hyphenat}   
\usepackage{siunitx}

\usepackage{amssymb,mathrsfs}
\usepackage{euscript}
\usepackage{dsfont}
\usepackage{booktabs}
\usepackage{url}

\usepackage[hidelinks]{hyperref}

\usepackage{tikz}
\usetikzlibrary{arrows.meta, positioning}

\newtheorem{theorem}{\textbf{Theorem}}[section]
\newtheorem{corollary}[theorem]{\textbf{Corollary}}
\newtheorem{definition}[theorem]{\textbf{Definition}}

\newtheorem{lemma}[theorem]{\textbf{Lemma}}

\newtheorem{proposition}[theorem]{\textbf{Proposition}}
\newtheorem{remark}[theorem]{\textbf{Remark}}

\newenvironment{proof}[1][Proof]{\noindent\textbf{#1.} }{\ \rule{0.5em}{0.5em}}

\renewcommand{\leq}{\leqslant}
\renewcommand{\le}{\leqslant}

\renewcommand{\ge}{\geqslant}
\renewcommand{\epsilon}{\varepsilon}

\newcommand{\F}{\mathbb{F}}

\newcommand{\fq}{\F_q}
\newcommand{\fqm}{\F_{q^m}}

\newcommand{\GL}{\mathsf{GL}}

\DeclareMathOperator{\rk}{\mathrm{rank}} 

\newcommand{\word}[1]{\boldsymbol{\mathrm{#1}}}

\newcommand{\cv}{\word{c}}

\newcommand{\ev}{\word{e}}

\newcommand{\gv}{\word{g}}

\newcommand{\sv}{\word{s}}
\newcommand{\uv}{\word{u}}
\newcommand{\vv}{\word{v}}

\newcommand{\xv}{\word{x}}
\newcommand{\yv}{\word{y}}

\newcommand{\zz}{\word{0}}

\newcommand{\lambdav}{\word{\lambda}}

\newcommand{\mat}[1]{\boldsymbol{\mathrm{#1}}}
\newcommand{\Am}{\mat{A}}
\newcommand{\Bm}{\mat{B}}
\newcommand{\Cm}{\mat{C}}

\newcommand{\Em}{\mat{E}}

\newcommand{\Gm}{\mat{G}}
\newcommand{ \Hm}{\mat{H}}
\renewcommand{\Im}{\mat{I}}

\newcommand{\Mm}{\mat{M}}
\newcommand{\Nm}{\mat{N}}
\newcommand{\Pm}{\mat{P}}
\newcommand{\Qm}{\mat{Q}}

\newcommand{\Um}{\mat{U}}

\newcommand{\Ym}{\mat{Y}}

\newcommand{\ZZ}{\mat{0}}

\newcommand{\CC}{\mathcal{C}}

\newcommand{\Tr}{\operatorname{Tr}}

\newcommand{\phiBvec}{\phi_B^{\mathrm{vec}}}
\newcommand{\phiBmat}{\phi_B^{\mathrm{mat}}}

\newcommand{\Fold}{\operatorname{Fold}}
\newcommand{\Unfold}{\operatorname{Unfold}}
\newcommand{\Stab}{\operatorname{Stab}}

\newcommand{\wt}{\mathrm{wt}}

\newcommand{\Cmat}{\mathcal{C}_{\mathrm{mat}}}

\newcommand{\Ann}{\operatorname{Ann}}

\DeclareMathOperator{\End}{End}
\DeclareMathOperator{\Ker}{Ker}
\newcommand{\Lq}{\mathcal{L}_{q^m}}
\begin{document}

\title{Subcodes of Lambda-Gabidulin Codes for Compact-Ciphertext Cryptography}

\author{Freddy Lend\'e Metouk\'e,
        Herv\'e Tal\'e Kalachi,
        Hermann Tchatchiem Kamche,
        Ousmane Ndiaye,
        Sélestin Ndjeya          
\thanks{Freddy Lend\'e Metouk\'e is with the Department of Mathematics, Faculty of Science, University of
Yaounde I, Cameroon email : metoukefreddy@gmail.com.}
\thanks{Herv\'e Tal\'e Kalachi is with the Department of Computer Engineering, National Advanced School of Engineering of Yaound\'e, University of Yaounde I, Cameroon e-mail: hervekalachi@gmail.com.}
\thanks{Hermann Tchatchiem Kamche is with the Centre for Cybersecurity and Mathematical Cryptology, The University of Bamenda, Bamenda, Cameroon e-mail: hermann.tchatchiem@gmail.com.}
\thanks{Ousmane Ndiaye is with the Universit\'e Cheikh Anta Diop de Dakar, FST, DMI, LACGAA, Senegal e-mail: ousmane3.ndiaye@ucad.edu.sn.}
\thanks{Sélestin Ndjeya is with the Department of Mathematics, Higher Teacher Training College, University of
Yaounde I, Cameroon e-mail : ndjeyas@yahoo.fr.}
}



\maketitle

\begin{abstract}
This paper investigates subcodes of lambda-Gabidulin codes, viewed as
rank-metric analogues of generalized Reed--Solomon codes, and their
applications to compact-ciphertext cryptosystems. We first analyze subspace and
generalized subspace subcodes of lambda-Gabidulin codes and relate them to
corresponding subcodes of classical Gabidulin codes through coordinate-wise
scaling. This relation yields cardinality bounds and structural properties for
these families. When the extension degree equals the code length, we further
characterize Gabidulin subspace subcodes in terms of linearized polynomials,
which gives an explicit description of their encoding and dimension. We also
study the matrix images of these subcodes over the base field through their
stabilizer and annihilator algebras, showing that subspace restrictions may
preserve nontrivial algebraic invariants despite the loss of extension-field
linearity. Motivated by these results, we propose a generator-matrix-based
construction of random subcodes designed to avoid such
invariants.
This construction is then used to design McEliece-like and Niederreiter-like
encryption schemes in the MinRank setting. Among the parameter sets considered
in this work, the most compact ciphertexts are obtained from random subcodes of classical Gabidulin codes. At the 128-, 192-, and
256-bit security levels, the resulting $\mathsf{LGS}$-Niederreiter instances
achieve the smallest ciphertext sizes among the compared schemes, while
maintaining competitive public-key sizes.
\end{abstract}


\begin{IEEEkeywords}
Code-based cryptography, rank-metric codes, Gabidulin codes, lambda-Gabidulin codes, subcodes, MinRank problem, public-key encryption.
\end{IEEEkeywords}

\section{Introduction}\label{sec:intro}

The rank metric was first introduced by Delsarte in 1978 in the setting of matrix codes over finite fields~\cite{delsarte78}. A few years later, Gabidulin developed the extension-field viewpoint and introduced the family now known as Gabidulin codes, together with their main structural properties and a decoding algorithm~\cite{gab85}. These codes are maximum rank distance codes, since they attain the Singleton bound in the rank metric. Shortly thereafter, Gabidulin, Paramonov, and Tretjakov proposed the GPT cryptosystem \cite{GPT91}, the first public-key cryptosystem based on rank-metric codes and a rank-metric analogue of the McEliece construction~\cite{M78}.

One of the main attractions of the rank metric in code-based cryptography is that the best known generic decoding attacks remain significantly harder in practice than their counterparts in the Hamming metric for comparable parameters~\cite{puchinger2022}. As a result, rank-metric cryptosystems can operate with much smaller code parameters and often achieve substantially smaller public keys in practice~\cite{gadouleau2006,bartz2022}. In principle, rank-metric codes therefore offer an appealing route toward compact code-based encryption. These advantages, however, can only be realized if the algebraic structure of the secret code is sufficiently well hidden in the public key, and this masking problem turned out to be the main obstacle for early Gabidulin-based proposals.

This difficulty was already visible in the original GPT cryptosystem. Its main weakness was precisely the challenge of hiding the strong algebraic structure of the underlying Gabidulin code. Gibson's attack~\cite{gib95} showed that this structure could already be exploited in the original proposal, and most subsequent variants were later broken by Overbeck's attacks~\cite{O05,O08} and their extensions~\cite{otmani18,horlemann18,kalachi22}. In retrospect, this vulnerability is perhaps not entirely surprising. Gabidulin codes are rank-metric analogues of Reed--Solomon codes, and McEliece-type cryptosystems based on Reed--Solomon-like structures have also been shown to admit powerful structural attacks~\cite{CGGOT14}. By contrast, the original McEliece cryptosystem based on binary Goppa codes has largely resisted efficient structural cryptanalysis for several decades in practical parameter ranges, even though Goppa codes can be viewed as subfield subcodes of generalized Reed--Solomon codes~\cite{delsarte1975,weger2022}. This contrast naturally suggests investigating subcodes of Gabidulin codes as possible rank-metric counterparts of the passage from generalized Reed--Solomon codes to Goppa codes. This perspective has already motivated several works on subcodes of Gabidulin codes.

Gabidulin and Loidreau initiated the study of subspace and subfield subcodes in the rank metric~\cite{gabsub05,gabprop08}. Their work was later continued by Gabidulin and Pilipchuk, who emphasized the relevance of subspace subcodes to network coding~\cite{gabNWC13}. More recently, Liu \emph{et al.} studied random Gabidulin subcodes from the point of view of list decoding and showed that they can, with overwhelming probability, achieve list-decodability behavior close to that of random rank-metric codes near the relevant Gilbert--Varshamov bound~\cite{liuASC16}. Subcodes have also been considered directly for cryptographic design. Berger \emph{et al.}~\cite{bergerMC17} proposed a cryptosystem based on matrix Gabidulin subcodes and highlighted the possibility of obtaining very short ciphertexts, an especially desirable feature in communication-constrained settings. Later, Guo \emph{et al.}~\cite{24two} introduced Loidreau-type variants based on Gabidulin subcodes. Taken together, these contributions indicate that subcodes are a natural way to weaken visible algebraic structure while retaining efficient decoding, and they also show that compact ciphertexts can be a decisive advantage in practical rank-metric cryptography.

In 2019, Lau and Tan introduced lambda-Gabidulin codes as a generalized Gabidulin family, which may be viewed as rank-metric analogues of generalized Reed--Solomon codes, and proposed a McEliece-type cryptosystem based on them~\cite{LT19}. These codes enlarge the design space beyond classical Gabidulin codes and were proposed to mitigate the structural weaknesses exploited by known attacks. More recently, however, several parameter sets of such constructions were shown to be vulnerable to structural cryptanalysis~\cite{burle26}. Although some parameters still resist the attacks currently known, these developments naturally raise the question of whether one should study subcodes of lambda-Gabidulin codes rather than the full family itself. From this viewpoint, subcodes of lambda-Gabidulin codes may provide a more plausible rank-metric analogue of the transition from generalized Reed--Solomon codes to Goppa codes. This question is particularly appealing because it combines two objectives: weakening visible algebraic structure while preserving the practical efficiency advantages that make rank-metric cryptography attractive.

In this paper, we investigate subcodes of lambda-Gabidulin codes from both a structural and a cryptographic viewpoint. On the structural side, we relate generalized subspace subcodes of lambda-Gabidulin codes to suitable subcodes of classical Gabidulin codes, derive cardinality bounds, and introduce a simple generator-matrix-based construction of random \(\F_q\)-linear subcodes. We also study the matrix images of these subcodes through their stabilizer and annihilator algebras, leading to structural distinguishers and to criteria for identifying restricted subcode families that remain algebraically visible.
On the cryptographic side, we use random \(\F_q\)-linear subcodes of lambda-Gabidulin codes and of their classical Gabidulin specialization to design McEliece-like and Niederreiter-like encryption schemes in the MinRank setting. Among the resulting instantiations, the best ciphertext-size tradeoffs are obtained in the classical Gabidulin case, while the broader lambda-Gabidulin framework provides additional design flexibility. In particular, the proposed Niederreiter-like construction achieves very compact ciphertexts together with competitive public-key sizes, improving the ciphertext-size tradeoff of previous rank-metric subcode-based proposals such as those of Berger \emph{et al.}~\cite{bergerMC17} and Guo \emph{et al.}~\cite{24two}, and comparing favorably with the enhanced matrix-code framework of Aragon \emph{et al.}~\cite{aragonMR24}.

The remainder of the paper is organized as follows. Section~\ref{sec:prelims}
recalls the necessary background on rank-metric codes, subcodes, and
lambda-Gabidulin codes. Section~\ref{sec:sub_subcodeLGab} studies subcodes of
lambda-Gabidulin and Gabidulin codes and establishes their main structural
properties. Section~\ref{stab_ssc} investigates the stabilizer and annihilator
algebras of matrix images of subcodes and derives the corresponding structural
distinguishers. Section~\ref{sec:crypto} presents the proposed cryptographic
constructions. Section~\ref{sec_secur} analyzes structural and decoding
attacks. Section~\ref{sec_prop_cle} provides parameter sets and comparisons
with related schemes. Finally, Section~\ref{sec:conclusion} concludes the
paper.

\section{Preliminaries}\label{sec:prelims}
This section is structured as follows. We begin with a brief overview of the rank metric, followed by a presentation of Gabidulin and $\lambdav$-Gabidulin codes. We conclude with some properties of subcodes in the rank metric.
 
\subsection{Matrix and Vector Representations of Rank-Metric Codes}
The rank metric was originally introduced by Delsarte~\cite{delsarte78}
in the study of matrix codes over finite fields. Let $q$ be a power of a prime,
let $\fq$ be the finite field with $q$ elements, and let $\fqm$ be its extension of degree $m$.
Throughout this subsection, we fix an $\fq$-basis $B=(b_1,\dots,b_m)$
of $\fqm$.

\begin{definition}[Matrix code]
A \emph{matrix code} is an $\F_q$-linear subspace
$\Cmat$ of the $\F_q$-vector space $\F_q^{m\times n}$ of all
$m\times n$ matrices over~$\F_q$.
The \emph{rank weight} of a matrix $\Mm\in \F_q^{m\times n}$ is defined by
\[
\wt_{\mathrm{R}}(\Mm):=\rk(\Mm),
\]
and the associated \emph{rank distance} on $\F_q^{m\times n}$ is
\[
d(\Mm,\Nm):=\rk(\Mm-\Nm)
\qquad
\text{for all } \Mm,\Nm\in\F_q^{m\times n}.
\]
The \emph{dual matrix code} of $\Cmat$ is
\[
\Cmat^\perp
:=
\left\{
\Am\in\F_q^{m\times n}
\;\middle|\;
\Tr(\Am\Bm^\top)=0
\text{ for all } \Bm\in\Cmat
\right\},
\]
where $\Tr(\cdot)$ denotes the ordinary trace of a square matrix over~$\F_q$ and $\Bm^\top$ the transpose of the matrix $\Bm$.
\end{definition}

\begin{definition}[Vector rank-metric code]\label{def:vector_rank_metric}
A \emph{vector rank-metric code} of length $n$ over $\F_{q^m}$ is an $\F_{q^m}$-linear subspace
$\CC\subseteq \F_{q^m}^n.$
For a vector
$\cv=(c_1,\dots,c_n) \in \F_{q^m}^n$,
its \emph{rank weight} is
\[
\rk_{\mathrm{wt}}(\cv)
:=
\dim_{\F_q}\bigl(\langle c_1,\dots,c_n\rangle_{\F_q}\bigr),
\]
where $\langle c_1,\dots,c_n\rangle_{\F_q}$ denotes the $\F_q$-subspace of $\F_{q^m}$
generated by the coordinates of $\cv$.
The associated \emph{rank distance} on $\F_{q^m}^n$ is
\[
d(\cv,\cv')
=
\rk_{\mathrm{wt}}(\cv-\cv')
\qquad
\text{for all } \cv,\cv'\in\F_{q^m}^n.
\]

The \emph{dual code} of $\CC$ is
\[
\CC^\perp
:=
\left\{
\cv\in\F_{q^m}^n
\;\middle|\;
\sum_{i=1}^n c_i c_i'=0
\text{ for all } \cv'=(c_1',\dots,c_n')\in\CC
\right\}.
\]
\end{definition}

\paragraph{\bf Expansion Maps}
Every element $x\in\fqm$ can be written uniquely as
\[
x=\sum_{j=1}^m x_j b_j,
\qquad x_j\in\fq.
\]
This defines the $\fq$-linear map
\[
\phi_B:\fqm\longrightarrow\fq^m,
\qquad
x\longmapsto (x_1,\dots,x_m).
\]
Extending componentwise, we obtain the \emph{vector expansion map}
\[
\phiBvec:\fqm^n\longrightarrow\fq^{mn},
\qquad
(x_1,\dots,x_n)\longmapsto
\bigl(\phi_B(x_1),\dots,\phi_B(x_n)\bigr),
\]
and the \emph{matrix expansion map}
\[
\phiBmat:\fqm^n\longrightarrow\fq^{m\times n},
\qquad
(x_1,\dots,x_n)\longmapsto
\bigl(\phi_B(x_1)^\top,\dots,\phi_B(x_n)^\top\bigr),
\]
whose columns are the $q$-ary expansions of the coordinates.

\begin{remark}\label{rem:rank_expansion}
For every $\xv\in\F_{q^m}^n$, one has
$\rk\bigl(\phiBmat(\xv)\bigr)=\rk_{\mathrm{wt}}(\xv).$
\end{remark}
In the sequel, we often write $\rk(\xv)$ in place of $\rk_{\mathrm{wt}}(\xv)$ when no confusion is possible.
Let $\CC\subseteq \F_{q^m}^n$ be an $\F_{q^m}$-linear code generated by
$\{\gv_1,\dots,\gv_k\}$. Since every coefficient in $\F_{q^m}$ expands uniquely on the basis
$B=(b_1,\dots,b_m)$, the code $\CC$, viewed as an $\F_q$-linear space, is generated by
\[
\{\, b_i \gv_j \mid 1\le i\le m,\; 1\le j\le k \,\}.
\]
Since $\phiBvec$ and $\phiBmat$ are $\F_q$-linear isomorphisms,
the codes $\phiBvec(\CC)$ and $\phiBmat(\CC)$ are generated by
\[
\{\phiBvec(b_i\gv_j)\}_{1\le i\le m,\;1\le j\le k}
\quad\text{and}\quad
\{\phiBmat(b_i\gv_j)\}_{1\le i\le m,\;1\le j\le k},
\]
respectively.

\paragraph{\bf Fold and Unfold}
Since the $\fq$-vector spaces $\fq^{mn}$ and $\fq^{m\times n}$ are canonically isomorphic,
we freely switch between vector and matrix representations depending on the context.
For
\[
\vv=(v_{11},\dots,v_{1m},v_{21},\dots,v_{2m},\dots,v_{n1},\dots,v_{nm})\in\fq^{mn},
\]
its \emph{folding} is
\[
\Fold(\vv)=
\begin{pmatrix}
v_{11} & v_{21} & \cdots & v_{n1}\\
\vdots & \vdots & \ddots & \vdots\\
v_{1m} & v_{2m} & \cdots & v_{nm}
\end{pmatrix}
\in\fq^{m\times n}.
\]
The inverse map is denoted by $\Unfold$ and, 
one can notice that for every $\xv\in\fqm^n$,
\[
\Fold\bigl(\phiBvec(\xv)\bigr)=\phiBmat(\xv),
\qquad
\Unfold\bigl(\phiBmat(\xv)\bigr)=\phiBvec(\xv).
\]

Figure \ref{fig} summarizes the connections between $\F_{q^m}^n, \F_q^{mn}, \F_q^{m\times n}$.  

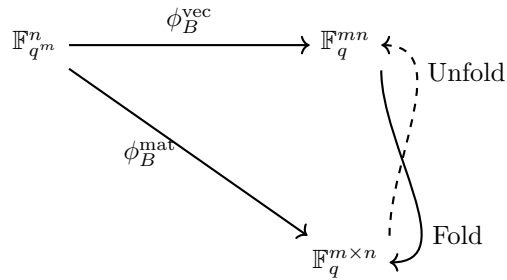
\begin{figure}[ht]
\centering
\begin{tikzpicture}[
    node distance=3.2cm,
    every node/.style={font=\small},
    arrow/.style={->, thick},
    foldarrow/.style={->, thick, black},
    unfoldarrow/.style={->, thick, black, dashed}
]

\node (Fqm) {$\mathbb{F}_{q^m}^n$};
\node (Fqmn) [right=of Fqm] {$\mathbb{F}_q^{mn}$};
\node (Fqmxn) [below=2.2cm of Fqmn] {$\mathbb{F}_q^{m\times n}$};

\draw[arrow] (Fqm) -- node[above] {$\phiBvec$} (Fqmn);
\draw[arrow] (Fqm) -- node[left] {$\phiBmat$} (Fqmxn);

\draw[foldarrow] (Fqmn.south east) to[out=-90, in=0] 
    node[right, pos=0.7] {\emph{$\Fold$}} 
    (Fqmxn.east);

\draw[unfoldarrow] (Fqmxn.north east) to[out=90, in=0] 
    node[right, pos=0.7] {\emph{$\Unfold$}} 
    (Fqmn.east);

\end{tikzpicture}
\caption{Relation between the vector and matrix representations over $\F_q$.}
\label{fig}
\end{figure}

\begin{remark}\label{rem: vec_unfold}
Let $\Mm,\Nm\in\F_q^{m\times n}$, and let $\mathrm{vec}(\cdot)$ denote the usual vectorization map.
By definition of $\Unfold$, one has
\[
\mathrm{vec}(\Mm)=\Unfold(\Mm)^\top.
\]
Using the classical identity
\[
\mathrm{vec}(\Mm)^\top\mathrm{vec}(\Nm)
=
\Tr(\Mm^\top\Nm),
\]
see~\cite{horn91}, it follows that
\[
\Unfold(\Mm) \Unfold(\Nm)^\top
=
\Tr(\Mm\Nm^\top).
\]
\end{remark}

\subsection{Gabidulin Codes and \texorpdfstring{$\lambdav$}{lambda}-Gabidulin Codes}

Gabidulin codes are rank-metric analogues of Reed--Solomon codes in the Hamming metric. They were introduced by Gabidulin in the extension-field setting and form a family of maximum rank distance codes~\cite{gab85}. We recall their standard definition.

\begin{definition}[Gabidulin code]
Let $n,m,k$ be positive integers such that $n \le m$ and $k \le n$, and let
$\gv=(g_1,\dots,g_n)\in\fqm^n$
be a vector whose coordinates are linearly independent over $\fq$.
The \emph{Gabidulin code} $\mathcal{G}(\gv,k)$ is the $\F_{q^m}$-linear code of length $n$ and dimension $k$ generated by the rows of
\begin{equation}\label{mat_gab}
\Gm =
\begin{pmatrix}
g_1 & g_2 & \cdots & g_n \\
g_1^{[1]} & g_2^{[1]} & \cdots & g_n^{[1]} \\
\vdots & \vdots & \ddots & \vdots \\
g_1^{[k-1]} & g_2^{[k-1]} & \cdots & g_n^{[k-1]}
\end{pmatrix},
\end{equation}
where $x^{[i]} := x^{q^i}$ denotes the $i$-th Frobenius power of $x$.
\end{definition}

Gabidulin codes are maximum rank distance (MRD) codes: their minimum rank distance is
\[
d=n-k+1.
\]
Moreover, the dual of a Gabidulin code is again a Gabidulin code~\cite{gab85}.

Inspired by the role of diagonal multipliers in generalized Reed--Solomon codes, Lau and Tan introduced the family of $\lambdav$-Gabidulin codes~\cite{LT19}.

\begin{definition}[\texorpdfstring{$\lambdav$}{lambda}-Gabidulin code]
Let $\gv=(g_1,\dots,g_n)\in\fqm^n$
be a vector whose coordinates are linearly independent over $\fq$, and let
$\lambdav=(\lambda_1,\dots,\lambda_n)\in(\F_{q^m}^*)^n$.
The \emph{$\lambdav$-Gabidulin code} $\mathcal{G}_{\lambdav}(\gv,k)$ is the $\F_{q^m}$-linear code generated by the rows of
\begin{equation}\label{mat_lam_gab}
\Gm_{\lambdav} =
\begin{pmatrix}
\lambda_1 g_1 & \lambda_2 g_2 & \cdots & \lambda_n g_n \\
\lambda_1 g_1^{[1]} & \lambda_2 g_2^{[1]} & \cdots & \lambda_n g_n^{[1]} \\
\vdots & \vdots & \ddots & \vdots \\
\lambda_1 g_1^{[k-1]} & \lambda_2 g_2^{[k-1]} & \cdots & \lambda_n g_n^{[k-1]}
\end{pmatrix}
=
\Gm\Delta,
\end{equation}
where $\Gm$ is given by~\eqref{mat_gab} and
\[
\Delta=\operatorname{Diag}(\lambda_1,\dots,\lambda_n).
\]
\end{definition}

If $\lambdav=(\lambda,\dots,\lambda)$ for some $\lambda \in \F_{q^m}^*$, then
\[
\mathcal{G}_{\lambdav}(\gv,k)=\mathcal{G}(\gv,k).
\]
Thus classical Gabidulin codes appear as a special case of $\lambdav$-Gabidulin codes.

Lau and Tan showed that the dual of a $\lambdav$-Gabidulin code is again a $\lambdav$-Gabidulin code~\cite[Proposition~1]{LT19}. Moreover, decoding a $\lambdav$-Gabidulin code reduces to decoding the associated Gabidulin code after multiplication by $\Delta^{-1}$. More precisely, if the Gabidulin code $\mathcal{G}(\gv,k)$ corrects up to $t$ rank-metric errors and if
\[
\delta=\rk(\lambdav^{-1})
\text{  with  }
\lambdav^{-1}=(\lambda_1^{-1},\dots,\lambda_n^{-1}),
\]
then the code $\mathcal{G}_{\lambdav}(\gv,k)$ can correct up to
$\left\lfloor \frac{t}{\delta}\right\rfloor$
rank errors~\cite[Proposition~4]{LT19}.

\subsection{Subcodes in the Rank Metric}\label{def_ss_code}

In this subsection, we recall several notions of subcodes in the rank metric that will be used throughout the paper.

\begin{definition}[Subcode]
Let $\mathcal{C}\subseteq \F_{q^m}^n$ be an $\F_{q^m}$-linear code, and let
$\mathcal{D}\subseteq \F_{q^m}^n$ be an $\F_q$-linear code.
We say that $\mathcal{D}$ is a \emph{subcode} of $\mathcal{C}$ if
$\mathcal{D}\subseteq \mathcal{C}.$
\end{definition}

The above definition is equivalent to say that, an $\F_q$-linear code $\mathcal{D}\subseteq \F_{q^m}^n$ is a subcode of $\mathcal{C}\subseteq \F_{q^m}^n$
if and only if there exists an $\F_q$-linear code
$\mathcal{C}'\subseteq \F_{q^m}^n$ such that
$\mathcal{D}=\mathcal{C}\cap \mathcal{C}'$.
Indeed, one may simply take $\mathcal{C}'=\mathcal{D}$.
This observation motivates the study of subcodes obtained by intersecting
$\mathcal{C}$ with suitable $\F_q$-linear subspaces of $\F_{q^m}^n$.

\begin{definition}[Subfield, subspace, and generalized subspace subcodes]
Let $\mathcal{C}\subseteq \F_{q^m}^n$ be an $\F_{q^m}$-linear code.
\begin{enumerate}
    \item The intersection
    $\mathcal{C}\cap \F_q^n$
    is called the \emph{subfield subcode} of $\mathcal{C}$ over $\F_q$~\cite{delsarteSFSC03,stichtenoth09}.

    \item Let $V\subseteq \F_{q^m}$ be an $\F_q$-vector subspace. Then
    $\mathcal{C}\cap V^n$
    is called the \emph{subspace subcode} of $\mathcal{C}$ over $V$,
    where
    $V^n=\underbrace{V\times\cdots\times V}_{n\text{ times}}$~\cite{gabprop08,hattori02,couvreur21}.

    \item Let $V_1,\dots,V_n$ be $\F_q$-vector subspaces of $\F_{q^m}$, and set
    $W=V_1\times\cdots\times V_n$.
    Then
 $\mathcal{C}\cap W$
    is called a \emph{generalized subspace subcode} of $\mathcal{C}$~\cite{ndiaye23,couvreur21}.
\end{enumerate}
\end{definition}

\begin{remark}\label{rem:subspace_limitation}
The restriction to $\F_q^n$ is of limited interest from the rank-metric viewpoint, since every nonzero vector of $\F_q^n$ has rank weight equal to $1$. More generally, if $V\subseteq \F_{q^m}$ is an $\F_q$-subspace of dimension $s$, then every word $\cv\in V^n$ satisfies
$\rk(\cv)\le s$.
Consequently, if $\mathcal{C}$ has minimum rank distance $d$, then the subspace subcode
$\mathcal{C}\cap V^n$
is trivial whenever $s<d$. This motivates the consideration of generalized subspace restrictions, where different coordinates may be restricted to different subspaces.
\end{remark}

Another classical way to construct subcodes is through the vertical concatenation of parity-check matrices, as in~\cite{bergerMC17,24two}. Let $\mathcal{C},\mathcal{C}'\subseteq \F_q^n$ be linear codes with parity-check matrices $\Hm$ and $\Hm'$, respectively. By standard duality arguments~\cite{lang12},
\[
(\mathcal{C}\cap \mathcal{C}')^\perp
=
\mathcal{C}^\perp+\mathcal{C}'^\perp.
\]
Therefore, the dual of $\mathcal{C}\cap \mathcal{C}'$ is generated by the rows of
\[
\widetilde{\Hm}=
\begin{pmatrix}
\Hm'\\ \hline
\Hm
\end{pmatrix}.
\]
This viewpoint underlies subcode constructions based on parity-check matrices.

\section{Generalized Subspace Subcodes of \texorpdfstring{$\lambdav$}{lambda}-Gabidulin Codes}
\label{sec:sub_subcodeLGab}

In this section, we study generalized subspace subcodes of $\lambdav$-Gabidulin codes. We first derive cardinality bounds for such subcodes and then, in the special case where the extension degree equals the code length, we give an algebraic characterization of subspace subcodes of Gabidulin codes.

\subsection{Cardinality Bounds for Generalized Subspace Subcodes of \texorpdfstring{$\lambdav$}{lambda}-Gabidulin Codes}

We begin by relating generalized subspace subcodes of $\lambdav$-Gabidulin codes to suitable generalized subspace subcodes of classical Gabidulin codes.
Throughout this subsection, let $n,m,k$ be positive integers such that $n\le m$, and let
$\mathcal{G}(\gv,k)\subseteq \F_{q^m}^n$
be a Gabidulin code with minimum rank distance
$d=n-k+1$.
Let
$\lambdav=(\lambda_1,\dots,\lambda_n)\in (\F_{q^m}^*)^n$
and define
$\Delta=\operatorname{Diag}(\lambda_1,\dots,\lambda_n)$,
so that
\[
\mathcal{G}_{\lambdav}(\gv,k)=\mathcal{G}(\gv,k)\Delta.
\]
For each $i\in\{1,\dots,n\}$, let $V_i\subseteq \F_{q^m}$ be an $\F_q$-subspace of dimension $s_i$, and set
\[
W:=\prod_{i=1}^n V_i \subseteq \F_{q^m}^n.
\]
We have the following lemma.

\begin{lemma}\label{lemma_fonda}
Let $\psi$ be the $\F_q$-linear map defined by
\[
\psi:\F_{q^m}^n\to \F_{q^m}^n,
\qquad
\cv\mapsto \cv\Delta^{-1}.
\]
Then $\psi$ restricts to an $\F_q$-linear bijection
\[
\psi_{\mid \mathcal{G}_{\lambdav}\cap W}:
\mathcal{G}_{\lambdav}\cap W
\longrightarrow
\mathcal{G}\cap \prod_{i=1}^n (\lambda_i^{-1}V_i).
\]
\end{lemma}

\begin{proof}
Since $\Delta\in \GL_n(\F_{q^m})$, the map $\psi$ is a bijection on $\F_{q^m}^n$.
Moreover,
$\psi(\mathcal{G}_{\lambdav})=\mathcal{G}$,
because $\mathcal{G}_{\lambdav}=\mathcal{G}\Delta$.
Finally, for $\cv=(c_1,\dots,c_n)\in W$, we have
\[
\psi(\cv)
=
(\lambda_1^{-1}c_1,\dots,\lambda_n^{-1}c_n)
\in
\prod_{i=1}^n (\lambda_i^{-1}V_i),
\]
and conversely $\psi^{-1}$ maps $\prod_{i=1}^n (\lambda_i^{-1}V_i)$ back to $W$.
Hence the restriction of $\psi$ yields the claimed bijection.
\end{proof}

\begin{theorem}\label{upper_bound_gene}
Assume that $\max_{1\le i\le n}\{s_i\}-d+1>0$.
Then the generalized subspace subcode $\mathcal{G}_{\lambdav}\cap W$ satisfies
\[
q^{\sum_{i=1}^n s_i - m(n-k)}
\le
\bigl|\mathcal{G}_{\lambdav}\cap W\bigr|
\le
q^{m(\max_{1\le i\le n}\{s_i\}-d+1)}.
\]

Moreover, if $s_1=\cdots=s_n=s$ and $m=n$, then the two bounds coincide, and the code
$\mathcal{G}_{\lambdav}(\gv,k)\cap W$
has $\F_q$-dimension
$k' = km + sm - m^2$.
\end{theorem}

\begin{proof}
By Lemma~\ref{lemma_fonda},
\[
\bigl|\mathcal{G}_{\lambdav}\cap W\bigr|
=
\left|\mathcal{G}\cap \prod_{i=1}^n (\lambda_i^{-1}V_i)\right|.
\]
The announced bounds then follow directly from \cite[Theorem~3]{ndiaye23} applied to the Gabidulin code $\mathcal{G}$.
\end{proof}

\begin{remark}\label{upper_bound}
If $V_1=\cdots=V_n=V$ and $\dim_{\F_q}(V)=s$, then the previous theorem specializes to
\[
q^{ns-m(n-k)}
\le
|\mathcal{G}_{\lambdav}\cap V^n|
\le
q^{m(s-d+1)},
\]
which recovers the corresponding bound for Gabidulin codes~\cite{gabprop08}.
\end{remark}

\subsection{\texorpdfstring{$q$}{q}-Polynomials and Subspace Subcodes}

We now recall some basic facts on $q$-polynomials and use them to study subspace subcodes.

\begin{definition}[\cite{ore1933theory}]
A \emph{$q$-polynomial} (or \emph{linearized polynomial}) over $\F_{q^m}$ is a polynomial of the form
\[
P(x)=\sum_{i=0}^d a_i x^{q^i},
\qquad a_i\in\F_{q^m}.
\]
Such a polynomial induces an $\F_q$-linear map on $\F_{q^m}$.
If $P(x)=\sum_{i=0}^d a_i x^{q^i}$ is nonzero with $a_d\neq 0$, its \emph{$q$-degree} is
\[
\deg_q(P)=d.
\]
By convention, $\deg_q(0)=-\infty$.
\end{definition}

Let $\mathcal{L}_{q^m}$ denote the set of $q$-polynomials over $\F_{q^m}$, equipped with addition and composition~$\circ$.  
The next lemma collects the properties of $q$-polynomials that will be needed later in the analysis of stabilizer algebras.

\begin{lemma}\label{lem:central_commute}
Let $V$ be an $\F_q$-subspace of $\F_{q^m}$ of dimension $s$, with $0<s<m$. Then the following properties hold.
\begin{enumerate}
    \item[(i)] There exists a nonzero $q$-polynomial $P\in\mathcal{L}_{q^m}$ of $q$-degree $s$ such that
    \[
    P(v)=0 \qquad \text{for all } v\in V.
    \]
    Moreover, $P$ is unique up to multiplication by a nonzero scalar in $\F_{q^m}$.

    \item[(ii)] Let $P$ be as in (i). Then there exists a $q$-polynomial $Q\in\mathcal{L}_{q^m}$ of $q$-degree $m-s$ such that
    \[
    Q\circ P=x^{q^m}-x.
    \]
    Moreover,
    \[
    P\circ Q=Q\circ P.
    \]
\end{enumerate}
\end{lemma}

\begin{proof}
(i) The existence of such a polynomial follows from~\cite{ore1933special}. 

(ii) Since $P$ vanishes on $V$ and $x^{q^m}-x$ vanishes on all of $\F_{q^m}$, \cite{ore1933special} implies that there exists a $q$-polynomial $Q\in\mathcal{L}_{q^m}$ of $q$-degree $m-s$ such that
\[
Q\circ P=x^{q^m}-x.
\]

Now, $x^{q^m}-x$ is a central element of $\mathcal{L}_{q^m}$ (see~\cite[Theorem~II-18]{mcdonald74}), that is,
\[
(x^{q^m}-x)\circ R=R\circ(x^{q^m}-x)
\qquad\text{for all }R\in\mathcal{L}_{q^m}.
\]
Applying this with $R=P$ gives
\[
(Q\circ P)\circ P=P\circ(Q\circ P).
\]
By associativity,
\[
(Q\circ P)\circ P=(P\circ Q)\circ P.
\]
Since $\mathcal{L}_{q^m}$ is an integral domain and $P\neq 0$, we may cancel the right factor $P$ and obtain
\[
Q\circ P=P\circ Q.
\]

\end{proof}

The evaluation of $q$-polynomials on a support vector gives an equivalent description of Gabidulin codes:
\[
\mathcal{G}(\gv,k)
=
\left\{
\bigl(P(g_1),\dots,P(g_n)\bigr)
\;\middle|\;
P\in\mathcal{L}_{q^m},\ \deg_q(P)<k
\right\}.
\]
The next proposition specializes this description to subspace subcodes in the case $m=n$.

\begin{proposition}\label{carac_ssc}
Assume that $m=n$. Let $\mathcal{G}=\mathcal{G}(\gv,k)$ be a Gabidulin code with support
$\gv=(g_1,\dots,g_n)\in\F_{q^m}^n$, whose coordinates are linearly independent over $\F_q$, and let $V\subseteq\F_{q^m}$ be an $\F_q$-subspace of dimension $s$.
Then there exists a $q$-polynomial $Q\in\mathcal{L}_{q^m}$ of $q$-degree $m-s$ such that
\[
\mathcal{G}\cap V^n
=
\left\{
\bigl((Q\circ A)(g_1),\dots,(Q\circ A)(g_n)\bigr)
\;\middle|\;
A\in\mathcal{L}_{q^m},\ \deg_q(A)<k-(m-s)
\right\}.
\]
\end{proposition}

\begin{proof}
By Lemma~\ref{lem:central_commute}, there exist $P,Q\in\mathcal{L}_{q^m}$ such that
\[
V=\ker(P),~ Q\circ P=x^{q^m}-x,~ \deg_q(P)=s,~ \deg_q(Q)=m-s.
\]

Let $\cv\in\mathcal{G}\cap V^n$. Then there exists a $q$-polynomial $F\in\mathcal{L}_{q^m}$ with $\deg_q(F)<k$ such that
\[
\cv=(F(g_1),\dots,F(g_n)).
\]
Since $\cv\in V^n=\ker(P)^n$, we have
\[
P(F(g_i))=0 \qquad \text{for all } i=1,\dots,n.
\]
Hence $(P\circ F)(g_i)=0$ for all $i$.

Because $m=n$ and the coordinates of $\gv$ are $\F_q$-linearly independent, they form an $\F_q$-basis of $\F_{q^m}$. Therefore $P\circ F$ vanishes on all of $\F_{q^m}$. By~\cite{ore1933special}, there exists $A\in\mathcal{L}_{q^m}$ such that
\[
P\circ F=A\circ(x^{q^m}-x).
\]
Since $x^{q^m}-x$ is central in $\mathcal{L}_{q^m}$,
\[
A\circ(x^{q^m}-x)=(x^{q^m}-x)\circ A=(Q\circ P)\circ A.
\]
Using Lemma~\ref{lem:central_commute}, we also have $P\circ Q=Q\circ P$, hence
\[
P\circ F=P\circ(Q\circ A).
\]
Since $\mathcal{L}_{q^m}$ is an integral domain and $P\neq 0$, we get
\[
F=Q\circ A.
\]
Moreover,
\[
\deg_q(A)=\deg_q(F)-\deg_q(Q)<k-(m-s).
\]
This proves one inclusion.

Conversely, let $A\in\mathcal{L}_{q^m}$ satisfy $\deg_q(A)<k-(m-s)$. Then
\[
\deg_q(Q\circ A)<k,
\]
so
\[
\bigl((Q\circ A)(g_1),\dots,(Q\circ A)(g_n)\bigr)\in\mathcal{G}.
\]
Furthermore,
\[
P((Q\circ A)(g_i))
=
(P\circ Q\circ A)(g_i)
=
((x^{q^m}-x)\circ A)(g_i)
=
0
\]
for all $i$, so this word belongs to $V^n$. Hence it lies in $\mathcal{G}\cap V^n$.
\end{proof}

\begin{remark}\label{rem:encoding_ssc}
\begin{enumerate}
    \item Proposition~\ref{carac_ssc} provides a direct description of the Gabidulin subspace subcode $\mathcal{G}\cap V^n$ in terms of the original support $\gv$ and $q$-polynomials, without passing through the parent-code framework of~\cite[Section~B]{gabprop08}.

    \item Moreover, if $k-m+s>0$, then the above parametrization shows that
    \[
    |\mathcal{G}\cap V^n| = q^{m(k-m+s)},
    \]
    and therefore
    \[
    \dim_{\F_q}(\mathcal{G}\cap V^n)=m(k-m+s).
    \]
    This recovers the corresponding dimension formula in~\cite{gabprop08}.

    \item Although $\mathcal{G}\cap V^n$ is in general only $\F_q$-linear, Proposition~\ref{carac_ssc} shows that it still retains a strong algebraic structure inherited from the extension field through the $q$-polynomial parametrization. In particular, the codewords are described by evaluations of compositions $Q\circ A$ with $A\in\mathcal{L}_{q^m}$, which reveals a residual extension-field structure behind the subspace restriction. This makes the existence of nontrivial structural invariants, and hence of potential distinguishers, less surprising.
\end{enumerate}
\end{remark}

\section{Stabilizer and Annihilator Algebra of Subspace Subcodes of Gabidulin Codes}\label{stab_ssc}

When an $\F_{q^m}$-linear code is expanded into a matrix code over $\F_q$, its hidden $\F_{q^m}$-linear structure may be revealed through non-trivial algebraic invariants, such as stabilizer \cite{aragonMR24}. This observation underlies the distinguisher of~\cite{aragonMR24} for Gabidulin matrix codes. To study this stabilizer algebra on the subcodes, we will first give a method to compute a dimension of this stabilizer algebra. Then, we will study the stabilizers of subspace subcodes and random subcodes.

\subsection{Stabilizer and Annihilator Algebras}

In this subsection, we recall the relevant definitions and derive a linear-algebraic characterization of the left stabilizer and left annihilator of a matrix code.

\begin{definition}
Let $\Cmat\subseteq \F_q^{m\times n}$ be a matrix code.
\begin{enumerate}
    \item The \emph{left stabilizer algebra} of $\Cmat$ is
    \[
    \Stab_L(\Cmat)
    :=
    \left\{
    \Pm\in \F_q^{m\times m}
    \;\middle|\;
    \Pm\Cm\in \Cmat
    \text{ for all } \Cm\in \Cmat
    \right\}.
    \]

    \item The \emph{left annihilator algebra} of $\Cmat$ is
    \[
    \Ann_L(\Cmat)
    :=
    \left\{
    \Am\in \F_q^{m\times m}
    \;\middle|\;
    \Am\Cm=\mathbf{0}
    \text{ for all } \Cm\in \Cmat
    \right\}.
    \]

    \item The \emph{right stabilizer algebra} and \emph{right annihilator algebra} are defined similarly, and are denoted by $\Stab_R(\Cmat)$ and $\Ann_R(\Cmat)$.
\end{enumerate}
\end{definition}

\begin{remark}
\begin{enumerate}
    \item One always has
    \[
    \Ann_L(\Cmat)\subseteq \Stab_L(\Cmat).
    \]
    Thus, the annihilator may be viewed as a substructure of the stabilizer algebra, consisting of linear transformations that act trivially on all codewords.

    \item As observed in~\cite{aragonMR24}, the existence of an $\F_{q^m}$-linear structure may be detected through a nontrivial right stabilizer algebra. Although subspace subcodes lose global $\F_{q^m}$-linearity, such codes may still exhibit nontrivial stabilizer algebras.

    \item For a random matrix code, one typically expects the annihilator algebras to
be trivial, namely reduced to $\{\ZZ\}$, and the stabilizer algebras to be
minimal, namely reduced to the scalar matrices:
\[
\Stab_L(\Cmat)=\{\alpha \Im_m \mid \alpha\in\F_q\},
\qquad
\Stab_R(\Cmat)=\{\alpha \Im_n \mid \alpha\in\F_q\},
\]
with overwhelming probability.
\end{enumerate}
\end{remark}

\begin{proposition}[Left stabilizer and annihilator via linear systems]\label{stab_ann_linear_system}
Let $\Cmat\subseteq \F_q^{m\times n}$ be an $\F_q$-linear matrix code of dimension $k'$, with ordered basis
\[
(\Gm_1,\dots,\Gm_{k'}).
\]
Let $(\Hm_1,\dots,\Hm_{mn-k'})$ be a basis of the dual code $\Cmat^\perp$.
\begin{enumerate}
    \item A matrix $\Am\in \F_q^{m\times m}$ belongs to $\Stab_L(\Cmat)$ if and only if, for all
    \(i=1,\dots,mn-k'\) and \(j=1,\dots,k'\),
    \[
    \Unfold(\Hm_i)(\Gm_j^\top\otimes \Im_m)\operatorname{vec}(\Am)=0.
    \]
    Furthermore, let $\Mm_S$ be the matrix obtained by stacking all row vectors
    \[
    \Unfold(\Hm_i)(\Gm_j^\top\otimes \Im_m),
    \qquad
    1\le i\le mn-k',\; 1\le j\le k',
    \]
    Then
    \[
    \dim_{\F_q}\bigl(\Stab_L(\Cmat)\bigr)=m^2-\rk(\Mm_S).
    \]

    \item A matrix $\Am\in \F_q^{m\times m}$ belongs to $\Ann_L(\Cmat)$ if and only if, for all
    \(j=1,\dots,k'\),
    \[
    (\Gm_j^\top\otimes \Im_m)\operatorname{vec}(\Am)=0.
    \]
    Additionally, if we define
    \[
    \Mm_S'=
    \begin{pmatrix}
    \Gm_1^\top\otimes \Im_m\\
    \vdots\\
    \Gm_{k'}^\top\otimes \Im_m
    \end{pmatrix}
    \in \F_q^{k'mn\times m^2},
    \]
    Then
    \[
    \dim_{\F_q}\bigl(\Ann_L(\Cmat)\bigr)=m^2-\rk(\Mm_S').
    \]
\end{enumerate}
\end{proposition}

\begin{proof}
We first consider the stabilizer. By definition,
\[
\Am\in \Stab_L(\Cmat)
\quad\Longleftrightarrow\quad
\Am\Gm_j\in \Cmat,
\text{ for all } j=1,\dots,k'.
\]
Since $\Cmat^\perp$ is the dual of $\Cmat$, this is equivalent to
\[
\Tr\bigl(\Hm_i(\Am\Gm_j)^\top\bigr)=0
\qquad
\text{for all } i=1,\dots,mn-k',\; j=1,\dots,k'.
\]
By Remark~\ref{rem: vec_unfold},
\[
\Tr\bigl(\Hm_i(\Am\Gm_j)^\top\bigr)
=
\Unfold(\Hm_i)\Unfold(\Am\Gm_j)^\top.
\]
Using the vectorization identity
\[
\operatorname{vec}(\Am\Gm_j)
=
(\Gm_j^\top\otimes \Im_m)\operatorname{vec}(\Am)
\qquad \text{(see~\cite[Lemma~4.3.1]{horn91})},
\]
we obtain
\[
\Tr\bigl(\Hm_i(\Am\Gm_j)^\top\bigr)
=
\Unfold(\Hm_i)(\Gm_j^\top\otimes \Im_m)\operatorname{vec}(\Am).
\]
Hence the stabilizer is exactly the solution space of the corresponding homogeneous linear system. The dimension formula follows from the rank--nullity theorem.

For the annihilator,
\[
\Am\in \Ann_L(\Cmat)
\quad\Longleftrightarrow\quad
\Am\Gm_j=\mathbf{0}
\text{ for all } j=1,\dots,k'.
\]
Vectorizing gives
\[
(\Gm_j^\top\otimes \Im_m)\operatorname{vec}(\Am)=0
\qquad \text{for all } j=1,\dots,k',
\]
which yields the matrix $\Mm_S'$. The dimension formula again follows from the rank--nullity theorem.
\end{proof}

Proposition~\ref{stab_ann_linear_system} shows that both the left stabilizer and the left annihilator of a matrix code can be computed through explicit homogeneous linear systems. Analogous constructions apply to the right stabilizer and right annihilator algebras.

\subsection{Stabilizer Algebra of Subspace Subcodes}

Subspace subcodes are only $\F_q$-linear, since restricting to an $\F_q$-subspace
$V\subseteq \F_{q^m}$ breaks the ambient $\F_{q^m}$-linearity. One might therefore
expect the algebraic invariants associated with the extension-field structure to
disappear. The next theorem shows that this is not necessarily the case. In
particular, subspace subcodes of Gabidulin codes, and more generally subspace
restrictions of linear codes over $\F_{q^m}$, may still retain large stabilizer
or annihilator algebras after expansion to matrix form. 

\begin{theorem}[Structural properties of stabilizer and annihilator algebras]
\label{thm:iterative_stab_ann}
Let $\gv\in\fqm^n$ be a support vector, let $\mathcal{G}(\gv,k)$ be a Gabidulin code, and let
$V\subseteq\fqm$ be an $\fq$-subspace of dimension $s$, with $0<s<m$.
Let $B$ be a fixed $\fq$-basis of $\fqm$. Then the following properties hold.
\begin{enumerate}
    \item If $m=n$, then the expanded subspace subcode
    \[
    \phiBmat\bigl(\mathcal{G}(\gv,k)\cap V^n\bigr)\subseteq \fq^{m\times n}
    \]
    admits a right stabilizer algebra containing an $\fq$-subalgebra isomorphic to $\fqm$. In particular,
    \[
    \dim_{\fq}\Bigl(\Stab_R\bigl(\phiBmat(\mathcal{G}(\gv,k)\cap V^n)\bigr)\Bigr)\ge m.
    \]

    \item For every $\fq$-linear code $\CC\subseteq\fqm^n$,
    \[
    \dim_{\fq}\Bigl(\Ann_L\bigl(\phiBmat(\CC\cap V^n)\bigr)\Bigr)\ge m(m-s).
    \]
    Consequently,
    \[
    \dim_{\fq}\Bigl(\Stab_L\bigl(\phiBmat(\CC\cap V^n)\bigr)\Bigr)\ge m(m-s)+1.
    \]
\end{enumerate}
\end{theorem}

\begin{proof}
\begin{enumerate}
    \item By Proposition~\ref{carac_ssc}, there exists a $q$-polynomial $Q\in\Lq$ of $q$-degree $m-s$
    such that codewords of $\mathcal{G}(\gv,k)\cap V^n$ are precisely the vectors of the form
    \[
    \bigl((Q\circ A)(g_1),\dots,(Q\circ A)(g_n)\bigr),
    \qquad
    A\in\Lq,\ \deg_q(A)<k-(m-s).
    \]
    Since $m=n$ and the coordinates of $\gv$ are $\fq$-linearly independent, the tuple
    $(g_1,\dots,g_n)$ is an $\fq$-basis of $\fqm$.

    For each $\alpha\in\fqm$, let $\Nm_\alpha\in\fq^{n\times n}$ be the matrix of the $\fq$-linear map
    \[
    L_\alpha:\fqm\longrightarrow\fqm,\qquad x\longmapsto \alpha x
    \]
    with respect to the basis $(g_1,\dots,g_n)$.

    Let
    \[
    \cv=\bigl((Q\circ A)(g_1),\dots,(Q\circ A)(g_n)\bigr)\in \mathcal{G}(\gv,k)\cap V^n.
    \]
    Since $Q\circ A$ is $\fq$-linear, one has
    \[
    \cv \Nm_\alpha
    =
    \bigl((Q\circ A)(\alpha g_1),\dots,(Q\circ A)(\alpha g_n)\bigr)
    =
    \bigl((Q\circ A\circ L_\alpha)(g_1),\dots,(Q\circ A\circ L_\alpha)(g_n)\bigr).
    \]
    Moreover,
    \[
    \deg_q(A\circ L_\alpha)=\deg_q(A)<k-(m-s),
    \]
    so $\cv \Nm_\alpha\in \mathcal{G}(\gv,k)\cap V^n$.
    Since $\Nm_\alpha\in\fq^{n\times n}$, this implies
    \[
    \phiBmat(\cv)\Nm_\alpha\in \phiBmat\bigl(\mathcal{G}(\gv,k)\cap V^n\bigr).
    \]
    Hence
    \[
    \Nm_\alpha\in \Stab_R\bigl(\phiBmat(\mathcal{G}(\gv,k)\cap V^n)\bigr)
    \qquad\text{for every }\alpha\in\fqm.
    \]

    Finally, the map
    \[
    \fqm\longrightarrow \fq^{n\times n},
    \qquad
    \alpha\longmapsto \Nm_\alpha
    \]
    is an injective $\fq$-algebra morphism. Therefore,
    \[
    \{\Nm_\alpha\mid \alpha\in\fqm\}
    \]
    is an $\fq$-subalgebra of
    \[
    \Stab_R\bigl(\phiBmat(\mathcal{G}(\gv,k)\cap V^n)\bigr)
    \]
    isomorphic to $\fqm$. In particular,
    \[
    \dim_{\fq}\Bigl(\Stab_R\bigl(\phiBmat(\mathcal{G}(\gv,k)\cap V^n)\bigr)\Bigr)\ge m.
    \]

    \item Let
\[
\mathcal{T}_V
:=
\bigl\{
T\in\End_{\fq}(\fqm)
\;\big|\;
V\subseteq\Ker(T)
\bigr\},
\]
where \(\End_{\fq}(\fqm)\) denotes the \(\fq\)-algebra of \(\fq\)-linear endomorphisms of \(\fqm\).
Choose an $\fq$-basis $(e_1,\dots,e_m)$ of $\fqm$ adapted to $V$, namely such that
\[
V=\langle e_1,\dots,e_s\rangle_{\fq}.
\]
Then an endomorphism $T\in\End_{\fq}(\fqm)$ belongs to $\mathcal{T}_V$ if and only if
the first $s$ columns of its matrix in the basis $(e_1,\dots,e_m)$ are zero.
Hence
\[
\dim_{\fq}(\mathcal{T}_V)=m(m-s).
\]

For each $T\in\mathcal{T}_V$, let $\Mm_T\in\fq^{m\times m}$ be the matrix of $T$
with respect to the fixed basis $B$.
Let $\cv=(c_1,\dots,c_n)\in \CC\cap V^n$. Since $c_i\in V$ for all $i$, one has
\[
T(c_i)=0
\qquad\text{for all }i=1,\dots,n.
\]
Now, by definition of $\phiBmat$, the $i$-th column of $\phiBmat(\cv)$ is
$\phi_B(c_i)^\top$.
Since $\Mm_T$ is the matrix of the $\fq$-linear map $T$ in the basis $B$, the
$i$-th column of $\Mm_T\,\phiBmat(\cv)$ is
\[
\Mm_T\,\phi_B(c_i)^\top=\phi_B(T(c_i))^\top=\zz.
\]
Therefore every column of $\Mm_T\,\phiBmat(\cv)$ is zero, and thus
\[
\Mm_T\,\phiBmat(\cv)=\ZZ.
\]
Hence
\[
\Mm_T\in \Ann_L\bigl(\phiBmat(\CC\cap V^n)\bigr).
\]

Since the map
\[
\mathcal{T}_V\longrightarrow \fq^{m\times m},
\qquad
T\longmapsto \Mm_T
\]
is injective and $\fq$-linear, its image is an $\fq$-subspace of
\[
\Ann_L\bigl(\phiBmat(\CC\cap V^n)\bigr)
\]
of dimension $m(m-s)$. Therefore,
\[
\dim_{\fq}\Bigl(\Ann_L\bigl(\phiBmat(\CC\cap V^n)\bigr)\Bigr)\ge m(m-s).
\]

Let
\[
\Cmat:=\phiBmat(\CC\cap V^n).
\]
Since $\Ann_L(\Cmat)\subseteq \Stab_L(\Cmat)$, it remains to prove that
\[
\dim_{\fq}\bigl(\Stab_L(\Cmat)\bigr)\ge \dim_{\fq}\bigl(\Ann_L(\Cmat)\bigr)+1.
\]
If $\Cmat\neq\{\ZZ\}$, then $\Im_m\in\Stab_L(\Cmat)$ but $\Im_m\notin\Ann_L(\Cmat)$, so
\[
\dim_{\fq}\bigl(\Stab_L(\Cmat)\bigr)\ge \dim_{\fq}\bigl(\Ann_L(\Cmat)\bigr)+1\ge m(m-s)+1.
\]
If $\Cmat=\{\ZZ\}$, then
\[
\Stab_L(\Cmat)=\fq^{m\times m},
\]
so
\[
\dim_{\fq}\bigl(\Stab_L(\Cmat)\bigr)=m^2\ge m(m-s)+1,
\]
because $0<s<m$. 
\end{enumerate}
\end{proof}

\begin{corollary}[Generalized coordinate restrictions]
\label{cor:coord_restriction}
Let $V_1,\dots,V_n\subseteq\fqm$ be $\fq$-subspaces, and set
\[
W:=V_1\times\cdots\times V_n,
\qquad
V:=V_1+\cdots+V_n.
\]
Assume that $V\neq\fqm$, and let
\[
r:=\dim_{\fq}(V).
\]
Then, for every $\fq$-linear code $\CC\subseteq\fqm^n$,
\[
\dim_{\fq}\Bigl(\Ann_L\bigl(\phiBmat(\CC\cap W)\bigr)\Bigr)\ge m(m-r).
\]
Consequently,
\[
\dim_{\fq}\Bigl(\Stab_L\bigl(\phiBmat(\CC\cap W)\bigr)\Bigr)\ge m(m-r) +1.
\]
\end{corollary}

\begin{proof}
Since each $V_i$ is contained in $V$, one has
\[
W\subseteq V^n.
\]
Hence
\[
\CC\cap W\subseteq \CC\cap V^n.
\]
By Item~(2) of Theorem~\ref{thm:iterative_stab_ann},
\[
\dim_{\fq}\Bigl(\Ann_L\bigl(\phiBmat(\CC\cap V^n)\bigr)\Bigr)\ge m(m-r),
\]
because $\dim_{\fq}(V)=r<m$.
Now every matrix that annihilates $\phiBmat(\CC\cap V^n)$ also annihilates
$\phiBmat(\CC\cap W)$, since
\[
\phiBmat(\CC\cap W)\subseteq \phiBmat(\CC\cap V^n).
\]
Therefore,
\[
\Ann_L\bigl(\phiBmat(\CC\cap V^n)\bigr)
\subseteq
\Ann_L\bigl(\phiBmat(\CC\cap W)\bigr),
\]
and thus
\[
\dim_{\fq}\Bigl(\Ann_L\bigl(\phiBmat(\CC\cap W)\bigr)\Bigr)\ge m(m-r).
\]
The bound on the stabilizer follows from the inclusion
\[
\Ann_L(\Cmat)\subseteq \Stab_L(\Cmat)
\]
for every matrix code $\Cmat$.
The final claim is immediate.
\end{proof}

Theorem~\ref{thm:iterative_stab_ann} and Corollary~\ref{cor:coord_restriction}
show that some natural classes of restricted subcodes retain quantitatively
large stabilizer or annihilator algebras after expansion to matrix form.
Therefore, these families remain algebraically visible from the viewpoint of
structural distinguishers.

\section{Applications in Cryptography}\label{sec:crypto}

The structural results obtained in the previous section have direct implications
for cryptographic constructions based on rank-metric codes. In particular, the
presence of non-trivial stabilizer or annihilator algebras provides potential
distinguishers that must be avoided when designing public codes. In this section,
we exploit these observations to construct encryption schemes based on
$\lambdav$-Gabidulin codes, using carefully chosen $\F_q$-linear subcodes.

\subsection{Selection of Public Subcodes}

A central design requirement for the public code is to avoid residual algebraic
structure that could be exploited by structural distinguishers. In the present
setting, the results of the previous section show that this issue is naturally
captured by the stabilizer and annihilator algebras of the matrix image of the
chosen subcode. Consequently, the main objective is to select $\F_q$-linear
subcodes of $\mathcal{G}_{\lambdav}(\gv,k)$ whose expansion behaves as much as
possible like a random matrix code.

Recall that any $\F_q$-linear subcode $\mathcal{D}$ of
$\mathcal{G}_{\lambdav}(\gv,k)$ can be written in the form
\[
\mathcal{D}=\mathcal{K}\cap \mathcal{G}_{\lambdav}(\gv,k),
\]
for some $\F_q$-linear code $\mathcal{K}\subseteq \F_{q^m}^n$.
From a cryptographic viewpoint, the choice of $\mathcal{K}$ is therefore
crucial: depending on this choice, the resulting subcode may either retain
detectable algebraic invariants after expansion to matrix form or resemble a
generic $\F_q$-linear rank-metric code.

Several natural choices of $\mathcal{K}$ lead to families that remain too
structured for cryptographic use. First, if $\mathcal{K}$ is
$\F_{q^m}$-linear, then $\mathcal{D}$ is itself $\F_{q^m}$-linear, and its
matrix image inherits the stabilizer phenomenon classically associated with
extension-field linearity. Second, if $\mathcal{K}=V^n$ for some
$\F_q$-subspace $V\subseteq \F_{q^m}$, then $\mathcal{D}$ is a subspace
subcode, and Theorem~\ref{thm:iterative_stab_ann} shows that its matrix image
retains explicit algebraic invariants after expansion. More precisely, its right
stabilizer remains large in the Gabidulin case, while its left annihilator and
left stabilizer admit quantitative lower bounds in the general $\F_q$-linear
setting. Third, if $\mathcal{K}=V_1\times\cdots\times V_n$ is a generalized
subspace restriction and
\[
r:=\dim_{\F_q}(V_1+\cdots+V_n)<m,
\]
then Corollary~\ref{cor:coord_restriction} implies that the corresponding matrix
code has a left annihilator algebra of dimension at least $m(m-r)$, and hence a
nontrivial left stabilizer algebra as well. In all these cases, the public code
remains algebraically distinguishable after expansion.

These observations suggest excluding such structured families from the public
code design. By contrast, random $\F_q$-linear subcodes are not forced a priori
to exhibit the stabilizer or annihilator patterns induced by extension-field
linearity or coordinate-wise subspace restrictions. They therefore constitute
the most promising candidates for constructing public codes that are both
cryptographically less distinguishable and flexible enough for practical
instantiation.

This motivates the generator-matrix approach developed in the next subsection,
whose purpose is precisely to produce $\F_q$-linear subcodes with a more
random-looking matrix image.

\subsection{Toward Random-Looking $\F_q$-Linear Subcodes}

The discussion of the previous subsection shows that some natural families of
subcodes of $\lambdav$-Gabidulin codes remain algebraically visible after
$q$-ary expansion. This motivates the search for a more flexible construction
of $\F_q$-linear subcodes, with fewer \emph{a priori} structural constraints.
To this end, we describe a simple generator-matrix-based construction derived
from the $q$-ary image of an $\F_{q^m}$-linear code.

Let $\mathcal{C}\subseteq \F_{q^m}^n$ be an $\F_{q^m}$-linear code of dimension
$k$. Since $\phiBvec(\mathcal{C})$ is an $\F_q$-linear code of dimension $km$,
any $\F_q$-linear subcode $\mathcal{D}\subseteq \mathcal{C}$ gives rise, via
$\phiBvec$, to an $\F_q$-linear subspace of $\phiBvec(\mathcal{C})$. The next
proposition shows that all such subcodes can be obtained by left-multiplying a
generator matrix of $\phiBvec(\mathcal{C})$ by a full-rank matrix over $\F_q$.

\begin{proposition}\label{subcode}
Let $B$ be an $\F_q$-basis of $\F_{q^m}$, 
$\mathcal{C}\subseteq \F_{q^m}^n$ be an $\F_{q^m}$-linear code of dimension
$k$, and $\Gm^{\mathrm{vec}}\in \F_q^{km\times mn}$ be a generator matrix
of $\phiBvec(\mathcal{C})$. Let $\mathcal{D}\subseteq \F_{q^m}^n$ be an
$\F_q$-linear code of dimension $k' \leq km $, and let
$\widetilde{\Gm}\in \F_q^{k'\times mn}$ be a generator matrix of
$\phiBvec(\mathcal{D})$. Then $\mathcal{D}$ is a subcode of $\mathcal{C}$ if
and only if there exists a full-rank matrix $\Pm\in\F_q^{k'\times km}$
such that
$\widetilde{\Gm}=\Pm\,\Gm^{\mathrm{vec}}$.
\end{proposition}

\begin{proof}
Since $\phiBvec$ is an $\F_q$-linear isomorphism, one has
\[
\mathcal{D}\subseteq \mathcal{C}
\quad\Longleftrightarrow\quad
\phiBvec(\mathcal{D})\subseteq \phiBvec(\mathcal{C}).
\]
Now $\widetilde{\Gm}$ and $\Gm^{\mathrm{vec}}$ are generator matrices of
$\phiBvec(\mathcal{D})$ and $\phiBvec(\mathcal{C})$, respectively. Therefore,
the inclusion
$\phiBvec(\mathcal{D})\subseteq \phiBvec(\mathcal{C})$
holds if and only if each row of $\widetilde{\Gm}$ is an $\F_q$-linear
combination of the rows of $\Gm^{\mathrm{vec}}$. Equivalently, there exists a
matrix $\Pm\in\F_q^{k'\times km}$ such that
\[
\widetilde{\Gm}=\Pm\,\Gm^{\mathrm{vec}}.
\]
Since $\widetilde{\Gm}$ has rank $k'$, the matrix $\Pm$ must also have rank
$k'$.
\end{proof}

Proposition~\ref{subcode} provides a convenient parametrization of
$\F_q$-linear subcodes of $\mathcal{C}$ through the $q$-ary image
$\phiBvec(\mathcal{C})$. In particular, it allows one to prescribe any subcode
dimension
$1\le k'\le km$. This flexibility is useful in
comparison with parity-check-based constructions, for which the achievable
dimensions are typically more constrained and may depend on additional rank
conditions~\cite{bergerMC17,24two}.

From a cryptographic viewpoint, this construction is attractive because it makes
it possible to avoid, at least at the design level, the structured families
discussed in the previous subsection. Indeed, if $m$ does not divide $k'$, then the
resulting subcode cannot be $\F_{q^m}$-linear, since every
$\F_{q^m}$-linear code has $\F_q$-dimension divisible by $m$. Moreover, unlike
subspace subcodes or generalized subspace restrictions, the present approach
does not impose any explicit coordinate-wise restriction pattern. For this
reason, choosing $\Pm$ at random yields natural candidates for public codes
whose $q$-ary images may behave more like generic matrix codes.
More precisely, for a fixed value of $k'$, one may sample a full-rank matrix
$\Pm\in \F_q^{k'\times km}$
uniformly at random and define
$\widetilde{\Gm}=\Pm\,\Gm^{\mathrm{vec}}$.
This produces an $\F_q$-linear subcode of $\mathcal{C}$ whose matrix image can
then be analyzed using Proposition~\ref{stab_ann_linear_system}.

Our experiments further support this design choice. For several parameter sets, we fixed a parent $\lambda$-Gabidulin code and generated 500 subcodes using the proposed construction. For each generated subcode, we computed the dimensions of the left and right stabilizer algebras of its matrix image. In almost all tested instances, both stabilizers were trivial, that is, of dimension~$1$. More precisely, this occurred in all tested cases over $q\in \{8,16\}$ and in all but two tested cases over $q=2$, where the empirical probability that the left stabilizer is trivial remained $0.998$. Although these experiments were restricted to moderate parameter sizes, they nevertheless provide strong evidence that the proposed construction yields subcodes whose matrix images behave much more like random matrix codes than the structured subcode families discussed above. The SageMath codes used for our simulations are publicly available at \url{https://github.com/Freddy-Lende/Lambda-Gabidulin-and-LGS-Niederreiter-cryptosystem}.

\subsection{Instantiating the MinRank Encryption Frameworks}

We now instantiate the McEliece and Niederreiter encryption frameworks in the
MinRank setting introduced in~\cite[Sections~4 and~5]{aragonMR24}, using public
codes derived from random-looking $\F_q$-linear subcodes of
$\lambdav$-Gabidulin codes.

In both constructions, the secret key is based on a
$\lambdav$-Gabidulin code $\mathcal{G}_{\lambdav}(\gv,k)$ together with a
chosen $\F_q$-basis $B$ of $\F_{q^m}$. The public key is obtained from an
$\F_q$-linear subcode
\[
\mathcal{D}\subseteq \mathcal{G}_{\lambdav}(\gv,k)
\]
of dimension $k'$, constructed through Proposition~\ref{subcode} by selecting a
full-rank matrix $\Pm\in\F_q^{k'\times km}$ and forming
\[
\widetilde{\Gm}=\Pm\,\Gm_{\lambdav}^{\mathrm{vec}},
\]
where $\Gm_{\lambdav}^{\mathrm{vec}}$ is a generator matrix of the expanded code
$\phiBvec(\mathcal{G}_{\lambdav}(\gv,k))$.

A distinctive feature of our instantiations is that, once the public subcode
$\mathcal{D}$ has been constructed, no additional masking transformation is
applied to derive the public key. In contrast with subcode-based constructions
that rely on extra disguising mechanisms or enhanced matrix-code
transformations~\cite{bergerMC17,aragonMR24}, our public key is simply a direct
representation of the chosen subcode. More precisely, in the McEliece-like
variant, the public key is given by a basis of the matrix code
$\phiBmat(\mathcal{D})$, whereas in the Niederreiter-like variant, it is given
by a parity-check matrix of the vector code $\phiBvec(\mathcal{D})$. In this
respect, the proposed approach is closer in spirit to the original McEliece
paradigm, where the public matrix directly describes a Goppa code \cite{M78}.

We now describe these two instantiations in detail.

\subsubsection{\bf The \texorpdfstring{$\mathsf{LGS}$}{LGS}-McEliece Cryptosystem}

In what follows, we specialize the general construction described above to the
McEliece-like MinRank framework of~\cite[Fig.~2]{aragonMR24}. We refer to the
resulting scheme as the \emph{Lambda-Gabidulin Subcode McEliece-like
cryptosystem}, abbreviated as $\mathsf{LGS}$-McEliece. In this variant, the
public code is represented directly by a basis of the matrix code
$\phiBmat(\mathcal{D})$, where
$\mathcal{D}\subseteq \mathcal{G}_{\lambdav}(\gv,k)$
is an $\F_q$-linear subcode constructed as in
Proposition~\ref{subcode}.

\paragraph{\bf Key generation}
The key generation starts from a secret $\lambdav$-Gabidulin code
$\mathcal{G}_{\lambdav}(\gv,k)$ and an $\F_q$-basis $B$ of $\F_{q^m}$. A
full-rank matrix $\Pm\in\F_q^{k'\times km}$ is then sampled in order to define
an $\F_q$-linear subcode $\mathcal{D}$ of prescribed dimension $k'$. The public
key is obtained by computing a basis of the matrix image
$\phiBmat(\mathcal{D})$, while the secret key keeps the compact description of
the underlying $\lambdav$-Gabidulin code together with the basis $B$. The
public rank-weight bound is chosen so as to remain within the decoding
capability of the $\lambdav$-Gabidulin code. The corresponding procedure is
summarized in Algorithm~\ref{alg:gene_McE}.

\begin{algorithm}[h!]
\caption{\textnormal{Key Generation for $\mathsf{LGS}$-McEliece}}
\label{alg:gene_McE}
\DontPrintSemicolon
\KwIn{$n=m$, integers $k\ge 1$ and $q$ a prime power.}
\KwOut{A public key $PK$ and a secret key $SK$.}

Choose an integer $k'$ not divisible by \(m\) and such that $1\le k'<km$.\;

Construct a generator matrix $\Gm_{\lambdav}\in \F_{q^m}^{k\times n}$ of the
$\lambdav$-Gabidulin code $\mathcal{G}_{\lambdav}(\gv,k)$.\;

Set
\[
\delta=\rk(\lambdav^{-1}),
\qquad
t_{\mathrm{pub}}=\left\lfloor \frac{n-k}{2\delta}\right\rfloor,
\]
where $\lambdav^{-1}=(\lambda_1^{-1},\dots,\lambda_n^{-1})$.\;

Choose a random $\F_q$-basis $B$ of $\F_{q^m}$.\;

Compute a generator matrix $\Gm_{\lambdav}^{\mathrm{vec}}$ of
$\phiBvec(\mathcal{G}_{\lambdav}(\gv,k))$.\;

Sample a uniformly random full-rank matrix
$\Pm\in \F_q^{k'\times km}$ and set
\[
\widetilde{\Gm}=\Pm\,\Gm_{\lambdav}^{\mathrm{vec}}.
\]
This matrix generates the $\F_q$-linear subcode
$\phiBvec(\mathcal{D})\subseteq \phiBvec(\mathcal{G}_{\lambdav}(\gv,k))$.\;

Compute a basis $(\Gm_1,\dots,\Gm_{k'})$ of
$\phiBmat(\mathcal{D})$.\;

Set
\[
PK=(\Gm_1,\dots,\Gm_{k'},t_{\mathrm{pub}}),
\qquad
SK=(B,\gv,\lambdav,k).
\]
\end{algorithm}

\paragraph{\bf Encryption}
Encryption follows the standard McEliece paradigm in matrix-code form. The
plaintext is viewed as a coefficient vector
\[
\xv=(x_1,\dots,x_{k'})\in\F_q^{k'},
\]
which selects a codeword in the public matrix code generated by
$(\Gm_1,\dots,\Gm_{k'})$. A low-rank error matrix is then generated at random
and added to the resulting codeword in order to hide it. Since the public code
is given directly by a basis of $\phiBmat(\mathcal{D})$, encryption is
particularly simple and amounts to a linear combination in the public basis
followed by the addition of a rank-bounded perturbation. This is summarized in
Algorithm~\ref{alg:enc_McE}.

\begin{algorithm}[h!]
\caption{\textnormal{Encryption for $\mathsf{LGS}$-McEliece}}
\label{alg:enc_McE}
\DontPrintSemicolon
\KwIn{A plaintext $\xv=(x_1,\dots,x_{k'})\in \F_q^{k'}$ and
$PK=(\Gm_1,\dots,\Gm_{k'},t_{\mathrm{pub}})$.}
\KwOut{A ciphertext $\Ym\in \F_q^{m\times n}$.}

Generate an error matrix $\Em\in \F_q^{m\times n}$ such that
$\rk(\Em)\le t_{\mathrm{pub}}$.\;

Compute
\[
\Ym=\sum_{i=1}^{k'} x_i\Gm_i+\Em.
\]

\Return{$\Ym$.}
\end{algorithm}

\paragraph{\bf Decryption}
Decryption uses the secret information of the ambient
$\lambdav$-Gabidulin code rather than the public subcode itself. More
precisely, the received matrix is first mapped back to a vector of
$\F_{q^m}^n$ through the inverse expansion map with respect to the secret basis
$B$. The resulting word is then decoded in the secret code
$\mathcal{G}_{\lambdav}(\gv,k)$, whose decoding algorithm corrects up to
$t_{\mathrm{pub}}$ rank-weight errors, thereby recovering the transmitted
codeword of the subcode $\mathcal{D}$. Finally, the plaintext is obtained by
expressing the recovered matrix codeword in the public basis
$(\Gm_1,\dots,\Gm_{k'})$. The procedure is given in
Algorithm~\ref{alg:dec_McE}.

\begin{algorithm}[h!]
\caption{\textnormal{Decryption for $\mathsf{LGS}$-McEliece}}
\label{alg:dec_McE}
\DontPrintSemicolon
\KwIn{A ciphertext $\Ym\in \F_q^{m\times n}$, the public basis
$(\Gm_1,\dots,\Gm_{k'})$, and $SK=(B,\gv,\lambdav,k)$.}
\KwOut{A plaintext $\xv=(x_1,\dots,x_{k'})\in \F_q^{k'}$.}

Compute
\[
\yv=(\phiBmat)^{-1}(\Ym)\in \F_{q^m}^n.
\]

Decode $\yv$ in the code $\mathcal{G}_{\lambdav}(\gv,k)$ and recover the
nearest codeword $\cv\in \mathcal{G}_{\lambdav}(\gv,k)$.\;

Compute $\Cm=\phiBmat(\cv)$.\;

Recover $\xv=(x_1,\dots,x_{k'})$ by solving the linear system
\[
\Cm=\sum_{i=1}^{k'} x_i\Gm_i
\]
over $\F_q$.\;

\Return{$\xv$.}
\end{algorithm}

\subsubsection{\bf The \texorpdfstring{$\mathsf{LGS}$}{LGS}-Niederreiter Cryptosystem}

We now specialize the general construction described above to the
Niederreiter-like MinRank framework of~\cite[Fig.~3]{aragonMR24}. We refer to
the resulting scheme as the \emph{Lambda-Gabidulin Subcode
Niederreiter-like cryptosystem}, abbreviated as $\mathsf{LGS}$-Niederreiter.
In this variant, the public code is represented directly by a parity-check
matrix of the vector code $\phiBvec(\mathcal{D})$, where
$\mathcal{D}\subseteq \mathcal{G}_{\lambdav}(\gv,k)$ is an $\F_q$-linear
subcode constructed as in Proposition~\ref{subcode}.

\paragraph{\bf Key generation}
Key generation starts from a secret $\lambdav$-Gabidulin code
$\mathcal{G}_{\lambdav}(\gv,k)$ and an $\F_q$-basis $B$ of $\F_{q^m}$. A
full-rank matrix $\Pm\in\F_q^{k'\times km}$ is then sampled in order to define
an $\F_q$-linear subcode $\mathcal{D}$ of prescribed dimension $k'$. The public
key is obtained by computing a parity-check matrix of the vector image
$\phiBvec(\mathcal{D})$, while the secret key keeps the compact description of
the underlying $\lambdav$-Gabidulin code together with the basis $B$. The
public rank-weight bound is chosen so as to remain within the decoding
capability of the $\lambdav$-Gabidulin code. The corresponding procedure is
summarized in Algorithm~\ref{alg:gene_Nied}.

\begin{algorithm}[h!]
\caption{\textnormal{Key Generation for $\mathsf{LGS}$-Niederreiter}}
\label{alg:gene_Nied}
\DontPrintSemicolon
\KwIn{$n=m$, integers $k\ge 1$ and $q$ a prime power.}
\KwOut{A public key $PK$ and a secret key $SK$.}

Choose an integer $k'$ not divisible by $m$ and such that $1\le k'<km$ .\;

Construct a generator matrix $\Gm_{\lambdav}\in \F_{q^m}^{k\times n}$ of the
$\lambdav$-Gabidulin code $\mathcal{G}_{\lambdav}(\gv,k)$.\;

Set
\[
\delta=\rk(\lambdav^{-1}),
\qquad
t_{\mathrm{pub}}=\left\lfloor \frac{n-k}{2\delta}\right\rfloor,
\]
where $\lambdav^{-1}=(\lambda_1^{-1},\dots,\lambda_n^{-1})$.\;

Choose a random $\F_q$-basis $B$ of $\F_{q^m}$.\;

Compute a generator matrix $\Gm_{\lambdav}^{\mathrm{vec}}$ of
$\phiBvec(\mathcal{G}_{\lambdav}(\gv,k))$.\;

Sample a uniformly random full-rank matrix
$\Pm\in \F_q^{k'\times km}$ and set
\[
\widetilde{\Gm}=\Pm\,\Gm_{\lambdav}^{\mathrm{vec}}.
\]
This matrix generates the $\F_q$-linear subcode
$\phiBvec(\mathcal{D})\subseteq \phiBvec(\mathcal{G}_{\lambdav}(\gv,k))$.\;

Compute a parity-check matrix
\[
\widetilde{\Hm}\in \F_q^{(mn-k')\times mn}
\]
of $\phiBvec(\mathcal{D})$.\;

Set
\[
PK=(\widetilde{\Hm},t_{\mathrm{pub}}),
\qquad
SK=(B,\gv,\lambdav,k).
\]
\end{algorithm}

\paragraph{\bf Encryption}
Encryption follows the standard Niederreiter paradigm in vector form. The
information to be transmitted is represented by an error vector
$\ev\in \F_q^{mn}$
whose folded matrix representation has rank at most $t_{\mathrm{pub}}$. The
ciphertext is then obtained as the syndrome of $\ev$ with respect to the public
parity-check matrix $\widetilde{\Hm}$. Since the public code is given directly
through a parity-check description of $\phiBvec(\mathcal{D})$, encryption
reduces to a single syndrome computation. This is summarized in
Algorithm~\ref{alg:enc_Nied}.

\begin{algorithm}[h!]
\caption{\textnormal{Encryption for $\mathsf{LGS}$-Niederreiter}}
\label{alg:enc_Nied}
\DontPrintSemicolon
\KwIn{$PK=(\widetilde{\Hm},t_{\mathrm{pub}})$ and an error vector
$\ev\in \F_q^{mn}$ such that $\rk(\Fold(\ev))\le t_{\mathrm{pub}}$.}
\KwOut{A ciphertext $\sv\in \F_q^{mn-k'}$.}

Compute
$\sv=\ev\,\widetilde{\Hm}^{\top}$.

\Return{$\sv$.}
\end{algorithm}

\paragraph{\bf Decryption}
Decryption combines the public syndrome description with the secret decoding
algorithm of the ambient $\lambdav$-Gabidulin code. Given a ciphertext
\(\sv\), one first computes any vector \(\yv\in\F_q^{mn}\) having syndrome
\(\sv\) with respect to the public parity-check matrix \(\widetilde{\Hm}\).
This vector is then mapped back to \(\F_{q^m}^n\) using the inverse expansion
map associated with the secret basis \(B\). Decoding in the secret code
\(\mathcal{G}_{\lambdav}(\gv,k)\) yields the nearest codeword \(\cv\), and the
transmitted error is finally recovered by subtracting \(\cv\) and expanding
back over \(\F_q\). The corresponding procedure is summarized in
Algorithm~\ref{alg:dec_Nied}.

\begin{algorithm}[h!]
\caption{\textnormal{Decryption for $\mathsf{LGS}$-Niederreiter}}
\label{alg:dec_Nied}
\DontPrintSemicolon
\KwIn{A ciphertext $\sv\in \F_q^{mn-k'}$, the public parity-check matrix
$\widetilde{\Hm}$, and $SK=(B,\gv,\lambdav,k)$.}
\KwOut{An error vector $\ev\in \F_q^{mn}$.}

Find any $\yv\in \F_q^{mn}$ such that
$\yv\,\widetilde{\Hm}^{\top}=\sv$.

Compute
\[
\yv'=(\phiBvec)^{-1}(\yv)\in \F_{q^m}^n.
\]

Decode $\yv'$ in the code $\mathcal{G}_{\lambdav}(\gv,k)$ and recover the
nearest codeword $\cv\in \mathcal{G}_{\lambdav}(\gv,k)$.\;

Set
$\ev'=\yv'-\cv$.

Return
$\ev=\phiBvec(\ev')$.

\end{algorithm}

\section{Security Analysis} \label{sec_secur}

In this section, we formalize the computational problems underlying the
security of the proposed schemes and discuss their cryptographic relevance. We
distinguish two main attack directions. Structural attacks aim at recovering the
secret key, or an equivalent one, by exploiting residual algebraic properties of
the public code. By contrast, decoding attacks treat the public code as
essentially random and attempt to recover plaintexts through generic algebraic
techniques, which in our setting reduce to solving instances of the MinRank
problem.

\subsection{Structural Attacks}

Structural attacks aim at exploiting residual algebraic properties of the
public code in order to distinguish it from a random matrix code or to recover
the secret key, or an equivalent one. Following the viewpoint of
\cite[Definitions~17 and~18]{aragonMR24}, we distinguish between a
distinguishing task and a search task adapted to the present setting.

\begin{definition}[$\mathsf{LGS}$--Distinguishing Problem]\label{def:Distinguishing_problem}
Let $\Cmat \subseteq \F_q^{m \times n}$ be a matrix code of dimension $k'$.
The $\mathsf{LGS}$--Distinguishing Problem consists in deciding, with
non-negligible advantage, whether $\Cmat$ is a subcode of the matrix image of a
$\lambdav$-Gabidulin code or a uniformly random matrix code of the same
dimension.
\end{definition}

\begin{definition}[$\mathsf{LGS}$--Search Problem]\label{def:search_problem}
Let $m,n,k'$ be positive integers with $n\le m$, and 
$\Cmat \subseteq \F_q^{m \times n}$ be a matrix code of dimension $k'$.
The $\mathsf{LGS}$--Search Problem consists in recovering, if possible, an
$\F_q$-basis $B$ of $\F_{q^m}$ and vectors
$\gv,\lambdav \in \F_{q^m}^n$ such that
\[
\Cmat \subseteq \phiBmat\!\left(\mathcal{G}_{\lambdav}(\gv,k)\right).
\]
More generally, one may also regard the $\mathsf{LGS}$--Search Problem as that of recovering any
equivalent secret key tupple $ \left( B^*, \gv^*, \lambdav^*, k \right)$ yielding the same public code. That is to say $\Cmat \subseteq \phi_{B^*}^{\mathrm{mat}}\!\left(\mathcal{G}_{\lambdav^*}(\gv^*,k)\right)$.
\end{definition}
A successful structural attack would then reveal the hidden
$\lambdav$-Gabidulin structure underlying the public subcode, together with the
expansion map induced by the secret basis $B$, or at least enough information to
derive an equivalent secret representation. This motivates the following attack
scenarios.

\paragraph{\bf Stabilizer-Based Attacks and Design Choices}

Stabilizer-based techniques provide efficient distinguishers between structured
matrix codes of Gabidulin type and generic random matrix codes. Their
effectiveness stems from the presence of non-trivial stabilizer algebras, which
typically do not occur for random matrix codes but naturally arise in several
highly structured families~\cite{aragonMR24,porwal25}.

Although subspace subcodes are only $\F_q$-linear,
Theorem~\ref{thm:iterative_stab_ann} shows that they may nevertheless admit a
non-trivial stabilizer algebra. Likewise,
Corollary~\ref{cor:coord_restriction} establishes that certain generalized
subspace restrictions also possess detectable stabilizers.

These observations indicate that a non-trivial stabilizer algebra should be
viewed as a structural weakness. For this reason, the public subcodes used in
our cryptosystem are chosen so as to avoid such configurations. More precisely,
Proposition~\ref{stab_ann_linear_system} provides an effective test to detect
whether a subcode generated at random through Proposition~\ref{subcode} admits a
non-trivial stabilizer algebra. Whenever such a stabilizer is detected, the
subcode is discarded and regenerated. This simple filtering step prevents
stabilizer-based distinguishing attacks against the proposed construction.

\paragraph{\bf Overbeck-Like Distinguisher}
When \(\mathcal{C}_{\mathrm{pub}}\) is an \(\F_q\)-linear subcode of
$\phiBmat\!\left(\mathcal{G}(\gv,k)\right)$,
that is, when the ambient secret code is a classical Gabidulin code, the
Overbeck-like argument of~\cite[Sec.~6.1.4]{aragonMR24} applies. More
precisely, using~\cite[Lemma~2]{aragonMR24}, one obtains, as in
\cite[Proposition~7 and Corollary~1]{aragonMR24}, a matrix code
$\mathcal{U}\subseteq \F_q^{m\times m}$
such that
\[
\dim\bigl(\mathcal{C}_{\mathrm{pub}}+\Um\mathcal{C}_{\mathrm{pub}}\bigr)
\le m(n-1),
 \text{ for all } \Um\in\mathcal{U}.
\]
By contrast, if \(\mathcal{C}_{\mathrm{pub}}\) is a random matrix code of
dimension \(k'\), one expects
\[
\dim\bigl(\mathcal{C}_{\mathrm{pub}}+\Um\mathcal{C}_{\mathrm{pub}}\bigr)
=
\min(mn,2k')
\]
with high probability, and in particular \(mn\) when \(mn\le 2k'\). This yields
a structural distinguisher in the Gabidulin case.

However, as explained in~\cite[pp.~22--23]{aragonMR24}, turning this property
into an explicit algebraic attack leads to a bilinear system with roughly
\(\Theta(m^2)\) unknowns and \(\Theta(m^2)\) equations, which remains out of
reach for the parameter ranges considered in this work. Therefore, although
this Overbeck-like property provides a formal structural distinguisher when
\(\mathcal{C}_{\mathrm{pub}}\subseteq \phiBmat(\mathcal{G}(\gv,k))\), it does
not currently lead to a practical attack. The general \(\lambdav\)-Gabidulin
case is not covered by this argument.

\paragraph{\bf Generator-Completion Attack}
We now formalize a structural attack in which the adversary attempts to
complete a generator matrix of the public subcode into a generator matrix of the
hidden expanded code. We first record the following lemma. 
\begin{lemma}\label{lem:qary-information-set-general}
Let $\mathcal{C}\subseteq \F_{q^m}^n$ be an $\F_{q^m}$-linear code of dimension
$k$, and let $I\subseteq \{1,\dots,n\}$ be an information set of
$\mathcal{C}$. Then the $km$ $q$-ary coordinate positions corresponding to the expansion of the
coordinates indexed by $I$ form an information set of the expanded code
$\phiBvec(\mathcal{C})\subseteq \F_q^{mn}$.
\end{lemma}

\begin{proof}
Let
\[
\pi_I:\mathcal{C}\longrightarrow \F_{q^m}^{k}
\]
denote the projection onto the coordinates indexed by $I$. Since $I$ is an
information set of $\mathcal{C}$, the map $\pi_I$ is an $\F_{q^m}$-linear
isomorphism.

Let
\[
\phi_{B,I}^{\mathrm{vec}}:\F_{q^m}^{k}\longrightarrow \F_q^{km}
\]
be the coordinatewise $q$-ary expansion map with respect to the basis $B$.
This is an $\F_q$-linear isomorphism.

Therefore, the composition
\[
\phi_{B,I}^{\mathrm{vec}}\circ \pi_I :
\mathcal{C}\longrightarrow \F_q^{km}
\]
is an $\F_q$-linear isomorphism. This means exactly that the $km$ positions  corresponding to the coordinates of the
$q$-ary expansion of the coordinates indexed by $I$ form an information set of
$\phiBvec(\mathcal{C})$.
\end{proof}



\begin{theorem}\label{thm:completion-normal-form}
Let $\mathcal{C}_{\mathrm{sec}}
=
\phiBvec\!\left(\mathcal{G}_{\lambdav}(\gv,k)\right)
\subseteq \F_q^{mn}$ and $\mathcal{F}\subseteq \mathcal{C}_{\mathrm{sec}}$  be an $\F_q$-linear
subcode of $\mathcal{G}_{\lambdav}(\gv,k)$, with dimension $k'$.

\begin{enumerate}
\item There exist a permutation matrix
\[
\Qm\in\F_q^{km\times km},
\]
and matrices
\[
\Am_1\in\F_q^{k'\times(km-k')},
\qquad
\Am_2\in\F_q^{k'\times(mn-km)}
\]
such that $\mathcal{F}$ admits a generator matrix of the form
\[
\widetilde{\Gm}
=
\begin{pmatrix}
(\Im_{k'}\; \Am_1)\,\Qm & \Am_2
\end{pmatrix}.
\]

\item For any fixed such normal form of $\widetilde{\Gm}$, there exists a unique
matrix
\[
\Am_3\in\F_q^{(km-k')\times(mn-km)}
\]
such that
\[
\widehat{\Gm}
=
\begin{pmatrix}
(\Im_{k'}\; \Am_1)\,\Qm & \Am_2\\
(\zz\; \Im_{km-k'})\,\Qm & \Am_3
\end{pmatrix}
\]
is a generator matrix of $\mathcal{C}_{\mathrm{sec}}$.
\end{enumerate}
\end{theorem}

\begin{proof}
Since $\{1,\dots,k\}$ is an information set of
$\mathcal{G}_{\lambdav}(\gv,k)$, Lemma~\ref{lem:qary-information-set-general}
implies that $\{1,\dots,km\}$ is an information set of $\mathcal{C}_{\mathrm{sec}}$. Therefore, if we denote by
$\pi:\F_q^{mn}\longrightarrow \F_q^{km}$
the projection onto the first $km$ coordinates, then 
$\pi_{\mid \mathcal{C}_{\mathrm{sec}}}:\mathcal{C}_{\mathrm{sec}}
\longrightarrow \F_q^{km}$
is an $\F_q$-linear isomorphism. It follows that its restriction to
$\mathcal{F}$ is injective and we have
$\dim_{\F_q}\pi(\mathcal{F})=k'$.
Therefore, $\pi(\mathcal{F})$ admits a generator matrix of the form
$(\Im_{k'}\; \Am_1)\,\Qm$,
for some permutation matrix $\Qm\in\F_q^{km\times km}$ and some
$\Am_1\in\F_q^{k'\times(km-k')}$.
Lifting back to $\mathcal{F}$ yields
\[
\widetilde{\Gm}
=
\begin{pmatrix}
(\Im_{k'}\; \Am_1)\,\Qm & \Am_2
\end{pmatrix},
\]
where $\Am_2\in\F_q^{k'\times(mn-km)}$, which proves the first claim.

For the second claim, let
\[
\uv_j=(\zz\; \ev_j)\,\Qm\in\F_q^{km},
\qquad 1\le j\le km-k',
\]
where $\ev_j$ is the $j$th standard basis vector of $\F_q^{km-k'}$. Since
$\pi_{\mid \mathcal{C}_{\mathrm{sec}}}:\mathcal{C}_{\mathrm{sec}}
\longrightarrow \F_q^{km}$
is bijective, each $\uv_j$ has a unique preimage
$\cv_j\in\mathcal{C}_{\mathrm{sec}}$.

Let the last $mn-km$ coordinates of $\cv_j$ form the $j$th row of a matrix
\[
\Am_3\in\F_q^{(km-k')\times(mn-km)}.
\]
Then all rows of
\[
\widehat{\Gm}
=
\begin{pmatrix}
(\Im_{k'}\; \Am_1)\,\Qm & \Am_2\\
(\zz\; \Im_{km-k'})\,\Qm & \Am_3
\end{pmatrix}
\]
belong to $\mathcal{C}_{\mathrm{sec}}$.

Moreover, the submatrix of $\widehat{\Gm}$ formed by its first $km$ columns has rank $km$. Hence the rows of $\widehat{\Gm}$ are linearly
independent. Since $\dim_{\F_q}(\mathcal{C}_{\mathrm{sec}})=km$,
$\widehat{\Gm}$ generates $\mathcal{C}_{\mathrm{sec}}$.

Uniqueness follows again from the bijectivity of
$\pi_{\mid \mathcal{C}_{\mathrm{sec}}}$. Once
\[
\widetilde{\Gm}
=
\begin{pmatrix}
(\Im_{k'}\; \Am_1)\,\Qm & \Am_2
\end{pmatrix}
\]
is fixed, each vector $(\zz\; \ev_j)\Qm$ has a unique preimage in
$\mathcal{C}_{\mathrm{sec}}$, hence each row of $\Am_3$ is uniquely determined.
\end{proof}

Theorem~\ref{thm:completion-normal-form} shows that, once a normal form of the
public subcode has been fixed, recovering the hidden expanded code amounts to
recovering the unique matrix $\Am_3$.

The attack therefore consists in fixing a generator matrix
$\widetilde{\Gm}
=
\begin{pmatrix}
(\Im_{k'}\; \Am_1)\,\Qm & \Am_2
\end{pmatrix}$
of the public subcode in the normal form of
Theorem~\ref{thm:completion-normal-form}, and then do an exhaustive search on $\Am_3$ in order to complete $\widetilde{\Gm}$ to have $\widehat{\Gm}$. Note that the number of candidate matrices for
$\Am_3\in\F_q^{(km-k')\times(mn-km)}$
is
$q^{(km-k')(mn-km)}$.
Consequently, an exhaustive completion-search attack requires
\[
q^{(km-k')(mn-km)}
\]
candidate tests, up to polynomial factors and the cost of the distinguisher
used to validate each candidate.

This generator-completion attack primarily recovers the hidden expanded code
\(
\mathcal{C}_{\mathrm{sec}}
\),
and therefore already distinguishes the public subcode from a random code.
Combined with a reconstruction procedure for q-ary images, such as those
considered in~\cite{bergerMC17}, it may also lead to an equivalent secret
description.

For distinguishing purposes, however, one can reduce the cost of the completion
search by puncturing the public subcode and keeping only $k+1$
extension-field coordinates. This is in the same spirit as the puncturing-based
dimension reduction used in~\cite{aragonMR24} in the context of structural
distinguishers.

\begin{corollary}\label{cor:completion-Distinguishing}
Let
$\mathcal{F}=\phiBvec(\mathcal{D})\subseteq
\mathcal{C}_{\mathrm{sec}}
=
\phiBvec\!\left(\mathcal{G}_{\lambdav}(\gv,k)\right)$
be the public subcode, of dimension $k'$. Let
$\mathcal{F}^{\mathrm{pun}}\subseteq \F_q^{m(k+1)}$
be the puncturing of $\mathcal{F}$ obtained by keeping only the first $k+1$
extension-field coordinates, i.e., the first $m(k+1)$ $q$-ary coordinates.
Then \(\dim_{\F_q}(\mathcal{F}^{\mathrm{pun}})=k'\), and
\(\mathcal{F}^{\mathrm{pun}}\) is an \(\F_q\)-linear subcode of
$\phiBvec\!\left(\mathcal{G}_{\lambdav'}(\gv',k)\right)$,
where \(\gv'\) and \(\lambdav'\) denote the truncations of \(\gv\) and
\(\lambdav\) to their first \(k+1\) entries. 

Consequently,
Theorem~\ref{thm:completion-normal-form} applies with \(n\) replaced by
\(k+1\), so that the completion search space drops
to
\[
\mathcal{C}_{\mathrm{dist}} = q^{m(km-k')}.
\]
\end{corollary}

\begin{proof}
By Lemma~\ref{lem:qary-information-set-general}, the first \(km\) $q$-ary
coordinates form an information set of
\[
\mathcal{C}_{\mathrm{sec}}
=
\phiBvec\!\left(\mathcal{G}_{\lambdav}(\gv,k)\right).
\]
Hence the projection onto the first \(m(k+1)\) $q$-ary coordinates is injective
on \(\mathcal{C}_{\mathrm{sec}}\), and therefore also on its subcode
\(\mathcal{F}\). It follows that
\[
\dim_{\F_q}(\mathcal{F}^{\mathrm{pun}})=\dim_{\F_q}(\mathcal{F})=k'.
\]

Moreover, puncturing \(\mathcal{G}_{\lambdav}(\gv,k)\) on the last \(n-k-1\)
extension-field coordinates yields the \(\lambdav'\)-Gabidulin code
\(\mathcal{G}_{\lambdav'}(\gv',k)\) of length \(k+1\). Thus
\(\mathcal{F}^{\mathrm{pun}}\) is an \(\F_q\)-linear subcode of
$\phiBvec\!\left(\mathcal{G}_{\lambdav'}(\gv',k)\right)$.
Applying Theorem~\ref{thm:completion-normal-form} with length \(k+1\) gives an
unknown block \(\Am_3\) of size
\[
(km-k')\times\bigl(m(k+1)-km\bigr)=(km-k')\times m,
\]
hence
$q^{m(km-k')}$
candidates.
\end{proof}

Thus, while the full completion attack is naturally viewed as a recovery attack
for the hidden expanded code, its punctured version provides a strictly cheaper
structural distinguisher.

\subsection{Decoding Attacks}

We begin by recalling the MinRank problem.

\begin{definition}[MinRank problem]
Let $q$ be a prime power and let $m,n,K,r \in \mathbb{N}$.
Given matrices
\[
\Am_1,\dots,\Am_K \in \mathbb{F}_q^{m \times n},
\]
the $\mathsf{MinRank}(q,m,n,K,r)$ problem consists in finding scalars
$x_1,\dots,x_K \in \mathbb{F}_q$ such that
\[
\rk \left(\sum_{i=1}^{K} x_i \Am_i\right) \leq r .
\]
\end{definition}

Decoding a matrix code can be reduced to an instance of the MinRank
problem~\cite{faugereMR08}. The corresponding decision problem was shown to be
NP-complete in~\cite{buss99}. The MinRank problem is therefore commonly used as
a hardness assumption in code-based cryptography~\cite{petzoldt15,beullens21,aragonMR24}.

In the present setting, once the public code is assumed to behave like a random
matrix code, message recovery attack reduces to solving a generic instance of the
$\mathsf{MinRank}(q,m,n,K,r)$ problem, with
\[
r=t_{\mathrm{pub}}
\qquad\text{and}\qquad
K=k'+1.
\]

We now review the main generic attacks considered in our estimates.

\paragraph{\bf The Kernel Attack}
The kernel attack introduced in~\cite{goubin00}
is one of the standard generic approaches for solving MinRank instances. Its
principle is to sample vectors in the right ambient space and to retain those
that lie in the kernel of the unknown matrix
\[
\sum_{i=1}^K x_i \Am_i
\]
of rank at most \(r\). Each such vector yields a linear equation in the unknown
coefficients, and once sufficiently many of them have been found, one obtains a
solvable linear system. Its complexity is
\[
\mathcal{C}_{\mathrm{ker}}
=
\mathcal{O}\!\left(q^{r\lceil K/m\rceil}K^{\omega}\right),
\]
where \(\omega\) denotes the linear algebra constant. In our estimates, we take \(\omega=2.38\), which is a convenient rounded value
for the Coppersmith--Winograd matrix-multiplication exponent~\cite{copper87}.

\paragraph{\bf The Support Minors Attack}
The support-minors approach, introduced in~\cite{bardet22}, provides an algebraic system for solving MinRank
instances. The complexity expression is taken under the rank assumption on the corresponding Macaulay matrix given in ~\cite{bardet22}, \cite{bros22}. It is then given by

\[
\mathcal{C}_{\mathrm{alg}}
=
\mathcal{O}\!\left(
\min\!\left(
E_b(q,m,n,K,r)\,U_b(q,m,n,K,r)^{\omega-1},
K(r+1)\,E_b(q,m,n,K,r)\,U_b(q,m,n,K,r)
\right)
\right),
\]

\

where $b$ is the smallest positive integer such that
\[
1 \le b < r+2
\qquad\text{and}\qquad
U_b(q,m,n,K,r)-1 \le E_b(q,m,n,K,r),
\]
and where:

If $q>2$, then 
\[
E_b(q,m,n,K,r)
:=
\sum_{i=1}^{b}(-1)^{i+1}
\binom{n}{r+i}
\binom{m+i-1}{i}
\binom{K+b-i-1}{b-i},
\]
and 
\[
U_b(q,m,n,K,r)
:=
\binom{K+b-1}{b}\binom{n}{r}.
\]

If $q=2$, then
\[
E_b(q,m,n,K,r)
:=
\sum_{j=1}^{b}\sum_{i=1}^{j}(-1)^{i+1}
\binom{n}{r+i}
\binom{m+i-1}{i}
\binom{K}{j-i},
\]
and
\[
U_b(q,m,n,K,r)
:=
\sum_{j=1}^{b}\binom{n}{r}\binom{K}{j}.
\]
Here, \(\binom{a}{b}\) denotes the binomial coefficient.

\paragraph{\bf The Hybrid Approach}
To improve the efficiency of attacks against the MinRank problem, a generic
hybrid strategy was proposed in~\cite{bardet23}. The idea is to reduce the
original instance to several smaller sub-instances that can be solved more
efficiently. Let $\mathcal{A}$ be an algorithm for solving the MinRank problem,
and let
\[
TC_{\mathcal{A}(q,m,n,K,r)}
\]
denote its time complexity on an instance with parameters \((q,m,n,K,r)\).
Using the hybrid technique of~\cite{bardet23}, the resulting complexity becomes
\[
TC_{\mathcal{A}\text{-}\mathrm{hybrid}(q,m,n,K,r)}
=
\min_{0\le a<\lceil K/m\rceil}
\left(
q^{ar}\cdot
TC_{\mathcal{A}(q,m,n-a,K-am,r)}
\right).
\]

\paragraph{\bf A Subsupport-Reduction Hybrid Approach}
A recent idea, introduced in~\cite{couvee25} in the context of the Rank
Decoding problem, is to use subsupport reduction as a preprocessing step.
We combine this reduction with the generic hybrid strategy recalled above. The
main idea is to guess a subsupport, that is, an $h$-dimensional subspace
contained in the rank support of the unknown low-rank solution. This guess
reduces the MinRank instance from parameters $(q,m,n,K,r)$ to
$(q,m,n-h,K,r-h)$ before applying a solver to the reduced instance.

Let \(\mathcal{A}\) be an algorithm for solving the MinRank problem, and let
$TC_{\mathcal{A}(q,m,n,K,r)}$
denote its time complexity on an instance with parameters \((q,m,n,K,r)\). In our MinRank adaptation, we model the average cost of the subsupport-reduction
step by
\[
\mathcal{C}_{\mathrm{sub\text{-}red}}
=
\mathcal{O}\!\left(
q^{h(n-r)}\cdot
TC_{\mathcal{A}(q,m,n-h,K,r-h)}
\right).
\]
Combining this reduction with the hybrid strategy of~\cite{bardet23} yields the
complexity
\[
\mathcal{C}_{\mathrm{new}\text{-}\mathsf{hyb}}
=
\min_{(a,h)\in B'(q,m,n,K,r)}
q^{h(n-r)+a(r-h)}
\cdot
TC_{\mathcal{A}(q,m,n-h-a,K-am,r-h)},
\]
where we set
\[
B'(q,m,n,K,r)
=
\{0,\ldots,\lceil K/m\rceil-1\}\times\{0,\ldots,r-1\}.
\]

When \(\mathcal{A}\) is instantiated with the support-minors attack
\(\mathsf{SM}\), by letting
$0\le a<\left\lceil \frac{K}{m}\right\rceil$, 
$0\le h<r$, and 
$1\le b< r-h+2$,
we define
\[
\begin{aligned}
B(q,m,n,K,r)
:=
\Bigl\{(a,b,h)\; ;\;
U_b(q,m,n-a-h,K-am,r-h)-1 \\
\le
E_b(q,m,n-a-h,K-am,r-h)
\Bigr\}.
\end{aligned}
\]
Then,
\[
\mathcal{C}_{\mathrm{new}\text{-}\mathrm{hyb}\text{-}\mathsf{SM}}
=
\mathcal{O}\!\left(
\min\{
\mathcal{C}_{\mathrm{new}\text{-}\mathrm{hyb}\text{-}\mathsf{SM1}},
\mathcal{C}_{\mathrm{new}\text{-}\mathrm{hyb}\text{-}\mathsf{SM2}}
\}
\right),
\]
where
\[
\begin{aligned}
\mathcal{C}_{\mathrm{new}\text{-}\mathrm{hyb}\text{-}\mathsf{SM1}}
=
\mathcal{O}~\!\Big(
\min_{(a,b,h)\in B(q,m,n,K,r)}
&q^{h(n-r)+a(r-h)} \\
&\cdot
E_b(q,m,n-h-a,K-am,r-h) \\
&\cdot
U_b(q,m,n-h-a,K-am,r-h)^{\omega-1}
\Big),
\end{aligned}
\]
and
\[
\begin{aligned}
\mathcal{C}_{\mathrm{new}\text{-}\mathrm{hyb}\text{-}\mathsf{SM2}}
=
\mathcal{O}~\!\Big(
\min_{(a,b,h)\in B(q,m,n,K,r)}
&q^{h(n-r)+a(r-h)}(K-am)(r-h+1) \\
&\cdot
E_b(q,m,n-h-a,K-am,r-h) \\
&\cdot
U_b(q,m,n-h-a,K-am,r-h)
\Big).
\end{aligned}
\]

When \(\mathcal{A}\) is instantiated with the kernel attack \(\mathsf{ker}\),
the resulting complexity is
\[
\mathcal{C}_{\mathrm{new}\text{-}\mathrm{hyb}\text{-}\mathsf{ker}}
=
\mathcal{O}~\!\left(
\min_{(a,h)\in B'(q,m,n,K,r)}
q^{h(n-r)+a(r-h)}
\,
q^{(r-h)\left\lceil (K-am)/m \right\rceil}
\,
(K-am)^{\omega}
\right).
\]

Finally, taking into account the distinguishing complexity
\(\mathcal{C}_{\mathrm{dist}}\) given in
Corollary~\ref{cor:completion-Distinguishing}, the overall work factor is
\begin{equation}\label{fin_compl}
\mathcal{C}_{\mathrm{f}}
=
\min\bigl\{
\mathcal{C}_{\mathrm{new}\text{-}\mathrm{hyb}\text{-}\mathsf{SM}},
\mathcal{C}_{\mathrm{new}\text{-}\mathrm{hyb}\text{-}\mathsf{ker}},
\mathcal{C}_{\mathrm{dist}}
\bigr\},
\end{equation}
where \(r=t_{\mathrm{pub}}\) and \(K=k'+1\).

\section{Proposed Parameters}\label{sec_prop_cle}

In this section, we propose parameter sets for the
$\mathsf{LGS}$-Niederreiter cryptosystem. Let
$\mathcal{G}(\gv,k)$ denote the Gabidulin code associated with the
$\lambdav$-Gabidulin code $\mathcal{G}_{\lambdav}(\gv,k)$, and let
\[
\delta=\rk(\lambdav^{-1}),
\qquad
\lambdav^{-1}=(\lambda_1^{-1},\dots,\lambda_n^{-1}).
\]
Then
\[
t_{\mathrm{pub}}=\left\lfloor \frac{n-k}{2\delta}\right\rfloor.
\]
The proposed parameter sets are selected according to the attack complexities
discussed in the previous section.

\subsection{Parameters for the \texorpdfstring{$\mathsf{LGS}$}{LGS}-Niederreiter Cryptosystem}

As in~\cite{aragonMR24}, we focus on the Niederreiter variant in order to
minimize the ciphertext size. The public key consists mainly of a systematic
parity-check matrix of the public code, and can therefore be stored using
\[
k'(mn-k')\log_2(q)\ \text{bits}.
\]
The ciphertext size is
\[
(mn-k')\log_2(q)\ \text{bits}.
\]
Tables~\ref{tab:128bit}, \ref{tab:192bit}, and \ref{tab:256bit} provide some parameter sets for \(128\)-, \(192\)-, and \(256\)-bit
security levels, together with the corresponding values of
\(\mathcal{C}_{\mathrm{f}}\) from~(\ref{fin_compl}). The two last columns of the tables contains the public key and ciphertext sizes.  

\begin{table}[!ht]
\centering

\caption{$\mathsf{LGS}$-Niederreiter: 128 bits security\label{tab:128bit}}
\begin{tabular}{@{} c c c c c c c r r S[table-format=3.2]@{}}
\toprule
 $q$ & $\delta$ & $m{=}n$ & $k$ & $k'$ & $t_\mathrm{pub}$ & $\mathcal{C}_{\mathrm{f}}$ & \textbf{pk (kB)} & \textbf{ct (B)}\\
\midrule
2 & 1 & 38 & 30 & 1125 & 4 & 131 & 44.86 & 40 \\
8 & 1 & 20 & 14 & 270 & 3 & 135 & 13.16 & 49 \\
2 & 1 & 34 & 24 & 800 & 5 & 131 & 35.60 & 45 \\
16 & 1 & 17 & 11 & 183 & 3 & 142 & 9.70 & 53 \\
2 & 1 & 32 & 18 & 564 & 7 & 137 & 32.43 & 58 \\
2 & 2 & 46 & 30 & 1360 & 4 & 132 & 128.52 & 95 \\
8 & 2 & 30 & 18 & 528 & 3 & 145 & 73.66 & 140 \\
\bottomrule
\end{tabular}
\end{table}

\begin{table}[!ht]
\centering

\caption{$\mathsf{LGS}$-Niederreiter: 192 bits security\label{tab:192bit}}
\begin{tabular}{@{} c c c c c c c r r S[table-format=3.2]@{}}
\toprule
 $q$ &$\delta$ & $m{=}n$ & $k$ & $k'$ & $t_\mathrm{pub}$ & $\mathcal{C}_{\mathrm{f}}$ & \textbf{pk (kB)} & \textbf{ct (B)}\\
\midrule
8 & 1 & 27 & 21 & 557 & 3 & 199 & 35.93 & 65 \\
16 & 1 & 23 & 17 & 381 & 3 & 213 & 28.19 & 74 \\
2 & 1 & 39 & 23 & 880 & 8 & 195 & 70.51 & 81 \\
16 & 1 & 21 & 11 & 221 & 5 & 213 & 24.31 & 110 \\
2 & 2 & 62 & 46 & 2822 & 4 & 196 & 360.51 & 128 \\
8 & 2 & 37 & 25 & 910 & 3 & 203 & 156.63 & 173 \\
\bottomrule
\end{tabular}

\end{table}

\begin{table}[!ht]
\centering
\caption{$\mathsf{LGS}$-Niederreiter: 256 bits security\label{tab:256bit}}
\begin{tabular}{@{} c c c c c c c r r S[table-format=3.2]@{}}
\toprule
 $q$ & $\delta$ & $m{=}n$ & $k$ & $k'$ & $t_\mathrm{pub}$ & $\mathcal{C}_{\mathrm{f}}$ & \textbf{pk (kB)} & \textbf{ct (B)}\\
\midrule
2 & 1 & 59 & 49 & 2881 & 5 & 259 & 216.07 & 75 \\
2 & 1 & 47 & 31 & 1434 & 8 & 260 & 138.92 & 97 \\
2 & 1 & 49 & 35 & 1695 & 7 & 257 & 149.58 & 89 \\
8 & 1 & 27 & 17 & 447 & 5 & 265 & 47.27 & 106 \\
16 & 1 & 24 & 14 & 324 & 5 & 276 & 40.82 & 126 \\
8 & 2 & 44 & 32 & 1388 & 3 & 269 & 285.23 & 206 \\
\bottomrule
\end{tabular}
\end{table}

\subsection{Comparison with Other Schemes}
In the following, we compare the public key and ciphertext sizes of our scheme with those of other encryption schemes. The comparisons are summarized in tables~\ref{tab:128bit_compa}, \ref{tab:192bit_compa}, and
\ref{tab:256bit_compa} where we compare representative instances of
$\mathsf{LGS}$-Niederreiter with several related cryptosystems  such as $\mathsf{EGMC}$-Niederreiter Cryptosystem  \cite{aragonMR24} based on expanded codes, MinRankPKE \cite{debris25}, Modification I of Loidreau's cryptosystem \cite{24two} and Cryptosystem Based on Equivalent of Subcodes of Gabidulin Matrix Codes (Cryptosystem $\mathsf{BESG}$)  \cite{berger17} at
\(128\)-, \(192\)-, and \(256\)-bit security levels.

\begin{table}[!ht]
\centering
\caption{Comparison with other cryptosystems, security: 128 bits}
\label{tab:128bit_compa}

\begin{tabular}{ l @{\hspace{8pt}} r @{\hspace{8pt}} r}
\toprule
 \textbf{Cryptosystems} & \textbf{pk (kB)} & \textbf{ct (B)} \\
\midrule
\textbf{$\mathsf{LGS}$-Niederreiter} & 44.86 & \textbf{40} \\
\textbf{$\mathsf{LGS}$-Niederreiter} & 9.70 & \textbf{53} \\
 Cryptosystem $\mathsf{BESG}$ \cite{bergerMC17} & 11.30 & 54  \\ 
  $\mathsf{EGMC}$-Niederreiter  \cite{vinc25} & 97.82 & 65 \\
  $\mathsf{Classic\;McEliece}$ \cite{bernstein17} & 261.12 & 96 \\
   Modification I of Loidreau's cryptosystem \cite{24two} & 3.67 & 105 \\
    $\mathsf{RQC}$-$\mathsf{Block}$-$\mathsf{NH}$-$\mathsf{MS}$-$\mathsf{AG}$~\cite{aragon2024} & 0.31\, & 1118 \\
 $\mathsf{BIKE}$~\cite{melchor171} & 1.54 & 1572 \\
  $\mathsf{RQC}$-$\mathsf{NH}$-$\mathsf{MS}$-$\mathsf{AG}$~\cite{bidoux23} & 0.42 & 2288 \\
 MinRankPKE \cite{debris25} & 14.70  & 14158 \\
\bottomrule
\end{tabular}
\end{table}

\begin{table}[!ht]
\centering
\caption{Comparison with other cryptosystems, security: 192 bits\label{tab:192bit_compa}}
\begin{tabular}{ l @{\hspace{8pt}} r @{\hspace{8pt}} r}
\toprule
 \textbf{Cryptosystems} & \textbf{pk (kB)} & \textbf{\;\;ct (B)} \\
\midrule
\textbf{$\mathsf{LGS}$-Niederreiter} & 35.93 & \textbf{65} \\
\textbf{$\mathsf{LGS}$-Niederreiter} & 28.19 & \textbf{74} \\
 $\mathsf{EGMC}$-Niederreiter \cite{vinc25} & 267.80 & 89 \\
 $\mathsf{Classic\;McEliece}$~\cite{bernstein17} & 524.16 & 156 \\
 Modification I of Loidreau's cryptosystem \cite{24two} & 5.48 & 457 \\
 $\mathsf{RQC}$-$\mathsf{Block}$-$\mathsf{NH}$-$\mathsf{MS}$-$\mathsf{AG}$~\cite{aragon2024} & 0.62 & 2278 \\
 $\mathsf{BIKE}$~\cite{melchor171} & 3.08 & 3024 \\
 $\mathsf{RQC}$-$\mathsf{NH}$-$\mathsf{MS}$-$\mathsf{AG}$~\cite{bidoux23} & 0.98 & 3753 \\
 MinRankPKE \cite{debris25} & 35.37 & 35365 \\

\bottomrule
\end{tabular}
\end{table}

\begin{table}[!ht]
\centering
\caption{Comparison with other cryptosystems, security: 256 bits\label{tab:256bit_compa}}
\begin{tabular}{ l @{\hspace{8pt}} r @{\hspace{8pt}} r}
\toprule
 \textbf{Cryptosystems} & \textbf{pk (kB)} & \textbf{\;\;ct (B)}\\
\midrule
 \textbf{$\mathsf{LGS}$-Niederreiter} & 216.07 & \textbf{75} \\
 \textbf{$\mathsf{LGS}$-Niederreiter} & 47.27 & \textbf{106} \\
 $\mathsf{EGMC}$-Niederreiter \cite{vinc25} & 274.29 & 139 \\
 $\mathsf{Classic\;McEliece}$~\cite{bernstein17} & 1 044.99 & 208 \\
 Cryptosystem $\mathsf{BESG}$ \cite{bergerMC17} & 136.58 & 265  \\
 Modification I of Loidreau's cryptosystem \cite{24two} & 8.70 & 622\\
 $\mathsf{BIKE}$~\cite{melchor171} & 5.12 & 5153 \\
 MinRankPKE \cite{debris25} & 62.89 & 62700 \\
\bottomrule
\end{tabular}
\end{table}

For all parameter sets considered here, $\mathsf{LGS}$-Niederreiter achieves
the smallest ciphertext size among the compared schemes. Moreover, some of the
proposed parameter sets provide a favorable public-key/ciphertext trade-off
compared with MinRankPKE, $\mathsf{EGMC}$-Niederreiter,
$\mathsf{BESG}$ cryptosystem, and $\mathsf{Classic\;McEliece}$.
These results indicate that $\mathsf{LGS}$-Niederreiter is a promising option
for communication-constrained settings.

\newpage
\section{Conclusion}\label{sec:conclusion}

In this paper, we studied subcodes of $\lambdav$-Gabidulin codes and analyzed
their suitability for cryptographic use. We showed that several natural
families of subcodes may retain visible algebraic invariants after $q$-ary
expansion, which makes them poor candidates for public-key design. This led us
to focus on random-looking $\F_q$-linear subcodes obtained through a simple
generator-matrix construction.

Using this construction, we proposed McEliece-like and Niederreiter-like
encryption schemes in the MinRank setting. We then analyzed their security
against structural and generic decoding attacks, including stabilizer-based
distinguishers, an Overbeck-like distinguisher, generator-completion attacks,
and the best currently known MinRank attacks. This analysis allowed us to
derive concrete parameters for the $\mathsf{LGS}$-Niederreiter cryptosystem.

The resulting scheme achieves very compact ciphertexts while keeping public-key
sizes competitive with those of related proposals. Overall, our results suggest
that $\F_q$-linear subcodes of $\lambdav$-Gabidulin codes provide a promising
framework for rank-metric encryption.

Several questions remain open. In particular, it would be valuable to obtain
sharper structural criteria for random $\F_q$-linear subcodes, especially to
better understand when their matrix images behave like random matrix codes.
Another important direction is to clarify the complexity of recovering a hidden
parent expanded code, or an equivalent description, from a public subcode.
More broadly, it would be interesting to determine whether the same approach
can be extended to other structured rank-metric primitives.
    
\bigskip

\bibliographystyle{IEEEtran}
\bibliography{codecrypto2}


\end{document}